\documentclass[3p]{elsarticle}
\usepackage{amsmath,amsfonts, amssymb}
\usepackage{amsthm}

\usepackage{fullpage}

\usepackage{graphicx}                          
\usepackage{setspace}  
\usepackage{caption}
\usepackage{listings}
\usepackage{tabularx}
\usepackage{xcolor} 
\usepackage{indentfirst}
\usepackage{subcaption}

\usepackage{multirow}
\usepackage{overpic}

\usepackage{verbatim}

\usepackage{enumerate}

\usepackage{tikz}

\usepackage[title]{appendix}

\usepackage{bm}

\usepackage{mhchem}

\usepackage[font=footnotesize,labelfont=bf]{caption}

\makeatletter
\@addtoreset{equation}{section}
\makeatother

\newtheorem{proposition}{Proposition}[section]

\newtheorem{remark}{Remark}[section]

\newcommand\dd{\mathrm{d}}
\newcommand\pp{\partial}
\newcommand\tr{\mathrm{tr}}

\newcommand\x{\bm{x}}

\newcommand\uvec{\bm{u}}
\newcommand\X{\mathbf{X}}
\newcommand\y{\bm{y}}
\newcommand\qvec{\bm{q}}
\newcommand\Qvec{\bm{Q}}

\newcommand\F{ \mathsf{F}}

\begin{document}

\title{A two species micro-macro model of wormlike micellar solutions and its maximum entropy closure approximations: An energetic variational approach}

\author[iit]{Yiwei Wang}
\ead{ywang487@iit.edu}

\author[cug]{Teng-Fei Zhang}
\ead{zhangtf@cug.edu.cn}

\author[iit]{Chun Liu}
\ead{cliu124@iit.edu}

\address[iit]{Department of Applied Mathematics, Illinois Institute of Technology, Chicago, IL 60616, USA}
\address[cug]{School of Mathematics and Physics, China University of Geosciences, Wuhan,430074, China}

\date{}


\begin{abstract}
  Wormlike micelles are self-assemblies of polymer chains that can break and recombine reversibly.
    In this paper, we derive a thermodynamically consistent two-species micro-macro model of wormlike micellar solutions by employing an energetic variational approach. The model 
    incorporates a breakage and combination process of polymer chains into the classical micro-macro dumbbell model of polymeric fluids in a unified variational framework. We also study different maximum entropy closure approximations to the new model by ``variation-then-closure'' and ``closure-then-variation'' approaches. 
  By imposing a proper dissipation in the coarse-grained level, the closure model, obtained by ``closure-then-variation'',  preserves the thermodynamical structure of both mechanical and chemical parts of the original system.  
  Several numerical examples show that the closure model can capture the key rheological features of wormlike micellar solutions in shear flows.
\end{abstract}

\maketitle


\section{Introduction}
Wormlike micelles, also known as ``living polymers'', are long, cylindrical aggregates of self-assembled surfactants that can break and recombine reversibly \cite{cates1987reptation}. 
There are substantial interests in studying wormlike micellar solutions for the purpose of fundamental research and industrial applications \cite{cates1990statics, cates2006rheology, jayaraman2003oscillations, smolka2003drop, yang2002viscoelastic}.
In particular, it has been observed that many wormlike micellar solutions exhibit shear banding, where the material splits into layers with different viscosities when undergoing strong shearing deformation \cite{olmsted2008perspectives}. Theoretically, shear banding is thought to arise from a non-monotone rheological constitutive curve of the shear stress versus the applied shear rate for steady homogeneous flow \cite{germann2014investigation, mohammadigoushki2019transient}. Understanding this unusual rheological behavior of wormlike micelles has been a focus of many theoretical and experimental studies \cite{vasquez2007network, germann2013nonequilibrium}. 

During the last couple of decades, a number of mathematical models have been proposed for wormlike micellar solutions \cite{adams2011transient, cates1987reptation, cates1990flow, dutta2018mechanistic, grmela2010mesoscopic, olmsted2000johnson, vasquez2007network}. Many theoretical models, such as the Johnson-Segalman model \cite{olmsted2000johnson} and Rolie-Poly model \cite{adams2011transient}, are one-species models, which didn't reflect the ``living'' nature of wormlike micelles. 
To account for the reversible breaking and combination of micellar chains, Cates proposed a reptation-reaction model, in which the reaction kinetics is introduced to account for the reversible breaking and combination process \cite{cates1987reptation,cates1990nonlinear}.
Inspired by Cates' seminal work, a two-species, scission-combination model for wormlike micellar solutions is proposed in \cite{vasquez2007network}, known as the VCM (Vasquez-Cook-McKinley) model. 
Although the VCM model was derived from a highly simplified discrete version of Cates' model \cite{germann2013nonequilibrium}, it can capture the key rheological properties of wormlike micellar solutions \cite{vasquez2007network, pipe2010wormlike, zhou2012multiple}. 
As pointed out in \cite{germann2016validation}, the VCM model is thermodynamically inconsistent, due to the assumption that the break rate depends on the velocity gradient explicitly.
The VCM model was later revisited into a thermodynamically consistent form \cite{germann2013nonequilibrium, germann2014investigation}, which is known as the GCB (Germann-Cook-Beris) model, by using generalized bracket approach \cite{beris1994thermodynamics}. Under the framework of GENERIC \cite{grmela1997dynamics, ottinger1997dynamics}, Grmela et al. formulate a mesoscopic tube model that includes the scission-recombination process, the reptation, Brownian relaxation and the diffusion, for wormlike micellar solutions \cite{grmela2010mesoscopic}, in which the wormlike micelles were modeled as different length chains composed of Hookean dumbbells.
Same framework can be used to derive several reduced models, including a VCM-type two-species model and a three-species model.



For many complex fluids, two-scale macro-micro models, which couple the evolution of the microscopic probability distribution function of polymeric molecules, with the macroscopic flow,  have been widely used to describe their dynamics \cite{le2012micro, li2007mathematical, lin2007micro}. 
In these models, the micro-macro interaction is coupled through a transport of the microscopic Fokker-Planck equation and the induced elastic stress tensor in the macroscopic equation. 
The competition between the kinetic energy and the multiscale elastic energies leads to different interesting hydrodynamical the rheological properties.
The goal of this paper is to extend such a micro-macro approach to model wormlike micellar solutions by incorporating it with the microscopic breaking and combination reaction kinetics.
 Following the setting in the VCM model \cite{vasquez2007network}, we represent the wormlike micelles by two species of dumbbells of different molecular weights respectively. Instead of constructing some empirical constitutive equation, we employ an energetic variational approach to derive the governing equation from an prescribed energy-dissipation law. 

We also study different maximum entropy closure approximations to the new micro-macro model. We adopt both ``closure-then-variation'' and ``variation-then-closure'' approaches. The first approach, which has been widely used in literature \cite{hyon2008maximum, wang2008crucial}, applies the maximum entropy closure at the PDE level, while the later approach first reformulates an energy-dissipation law at a coarse grained level and derives the closure system by a variation procedure \cite{doi2016principle, klika2019dynamic}. Due to the presence of reaction kinetics, these two approaches are not equivalent. Although the ``closure-then-variation'' approach can obtain a model satisfying a energy-dissipation property, the model fails to produce a non-monotone rheological constitutive curve of the shear stress versus the applied shear rate in steady homogeneous flows. In contrast, by formulating the dissipation part in the coarse-grained level properly, the ``closure-then-variation'' approach can result in a model that preserves the thermodynamical structures of the mechanical and chemical parts of the original system.
The resulting closure system takes the same form as the VCM \cite{vasquez2007network} and GCB models \cite{germann2013nonequilibrium}.
Numerical simulations show that the moment closure model can capture the key rheological features of wormlike micellar solutions as the VCM \cite{vasquez2007network} and GCB \cite{germann2013nonequilibrium} models.

The paper is organized as follows. In section 2, we formally derive the micro-macro model for wormlike micellar solutions by employing an energetic variational approach. A detailed investigation of maximum entropy closure approximations to the micro-macro model is presented in section 3. In section 4, we show the closure model, obtained by ``closure-then-variation'' can capture the key rheological features of wormlike micellar solutions in planar shear flows.


\section{Energetic variational formation of the new micro-macro model}

In this section, we employ an energetic variational approach to derive a thermodynamically consistent two-species micro-macro model for wormlike micellar solutions. 

\subsection{Energetic variational approach}

Originated from  pioneering works of Rayleigh \cite{strutt1871some} and Onsager \cite{onsager1931reciprocal, onsager1931reciprocal2}, 
the energetic variational approach (EnVarA) provides a general framework to 
derive the dynamics of a nonequilibrium thermodynamic system from a prescribed energy-dissipation law through two distinct variational processes: the Least Action Principle (LAP) and the Maximum Dissipation Principle (MDP) \cite{Giga2017,liu2009introduction}. 
The energy-dissipation law, which comes from the first and second laws of thermodynamics \cite{ericksen1992introduction, Giga2017}, can be formulated as
\begin{equation}
\frac{\dd}{\dd t} E^{\text{total}}(t) = - \triangle,
\end{equation}
for an isothermal closed system.
Here $E^{\text{total}}$ is the total energy, which is the sum of the Helmholtz free energy $\mathcal{F}$ and the kinetic energy $\mathcal{K}$; $\triangle$ is the rate of energy dissipation, which is equal to the entropy production in this case.
The LAP states that the dynamics of a Hamiltonian system is determined as a critical point of the action functional $\mathcal{A}(\x) = \int_0^T (\mathcal{K} - \mathcal{F}) \dd t$ with respect to $\x$ (the trajectory in Lagrangian coordinates for mechanical systems) \cite{arnol2013mathematical, Giga2017}, i.e., 
\begin{equation}
  \delta \mathcal{A} =  \int_{0}^T \int_{\Omega(t)} (f_{\text{inertial}} - f_{\text{conv}})\cdot \delta \x~  \dd \x \dd t.
\end{equation}
In the meantime, for a dissipative system, the dissipative force can be determined by minimizing the dissipation functional $\mathcal{D} = \frac{1}{2} \triangle$ with respect to the ``rate'' $\x_t$ in the linear response regime \cite{de2013non}, i.e.,
\begin{equation}
\delta \mathcal{D} = \int_{\Omega(t)} f_{\text{diss}} \cdot \delta \x_t~ \dd \x.
\end{equation}
This principle is known as Onsager's MDP \cite{onsager1931reciprocal, onsager1931reciprocal2}.
Thus, according to force balance (Newton's second law, in which the inertial force plays role of $ma$), we have
\begin{equation}\label{FB}
\frac{\delta A}{\delta \x} = \frac{\delta \mathcal{D}}{\delta \x_t}
\end{equation}
in Eulerian coordinates, which is the dynamics of the system. 
In the framework of EnVarA,  the dynamics of the system is totally determined by the energy-dissipation law and the kinematic relation, which shifts the main task of modeling complex nonequilibrium systems to the construction of energy-dissipation laws. 
The EnVarA framework has been proved to be a powerful tool to build up thermodynamically consistent mathematical models for many complicated system, especially those in complex fluids \cite{liu2009introduction, Giga2017}. 

Complex fluids are fluids with complicated rheological phenomena, arising from the interaction between the microscopic elastic properties and the macroscopic flow motions \cite{lin2007micro, liu2009introduction}.
A central problem in modeling complex fluids is to construct a constitutive relation, which links the stress tensor ${\bm \tau}$ and the velocity field $\nabla \uvec$ \cite{le2009multiscale}. 
Unlike a Newtonian fluid, there is no simple linear relation ${\bm \tau} = \mu \dot{\bm \gamma},$ where $\dot{\bm \gamma} = \frac{1}{2}( \nabla \uvec + \nabla \uvec^{\rm T})$ is the strain rate and $\mu$ is the viscosity, for complex fluids.
Instead of constructing an empirical constitutive equation that often takes the form of 
\begin{equation}
  \pp_t {\bm \tau} + (\uvec \cdot \nabla) {\bm \tau} = \bm{f}(\bm{\tau}, \nabla \uvec),
  \end{equation}
the EnVarA framework derives the constitutive relation from the giving energy-dissipation law through the variation procedure. Hence, the multiscale coupling and competition among multiphysics can be dealt with systematically.


As an illustration, we first give a formal derivation of a one-species incompressible micro-macro model of a dilute polymeric fluid by employing the EnVarA. A more detailed description to this model can be found in \cite{lin2007micro, Giga2017}. 
In this model, it is assumed that the polymeric fluid consists of beads joined by springs, and a molecular configuration is described by an end-to-end vector $\qvec \in \mathbb{R}^d$  \cite{bird1987dynamics, doi1988theory}.  At the microscopic level, 
the system is described
by a Fokker-Planck equation of the number distribution function $\psi(\x, \qvec, t)$ with a drift term depending on the macroscopic velocity $\uvec$. While the macroscopic motion of the fluid is described by a Navier-Stoke equation with an elastic stress depending on the $\psi(\x, \qvec, t)$.

To derive the dynamics of the system by the EnVarA, we need to introduce Lagrangian descriptions in both microscopic and macroscopic scales.
In the macroscopic domain $\Omega$, we define the flow map $\x(\X, t): \Omega \rightarrow \Omega$, where $\X$ are Lagrangian coordinates and $\x$ are Eulerian coordinates.  For fixed $\X$, $\x(\X, t)$ is the trajectory of a particle labeled by $\X$, while for fixed $t$, $\x(\X, t)$ is an orientation-preserving diffeomorphism between the initial domain to the current domain.
For a given flow map $\x(\X, t)$, we can define the associated velocity 
\begin{equation}
\uvec(\x(\X, t), t) = \frac{\dd}{\dd t} \x(\X, t),
\end{equation}
and the deformation tensor 
\begin{equation}
\widetilde{\F}(\x(\X, t), t) = \F(\X, t) = \nabla_{\X} \x(\X, t).
\end{equation}
Without ambiguity, we will not distinguish $\F$ and $\widetilde{\F}$ in the following. It is easy to verify that $\F(\x, t)$ satisfies the transport equation \cite{lin2007micro}
\begin{equation}\label{transport_F}
\F_t  + \uvec \cdot \nabla \F = \nabla \uvec \F,
\end{equation}
in Eulerian coordinates, where $\uvec \cdot \nabla \F$ stands for $u_k \pp_k F_{ij}$.
The deformation tensor $\F$ carries all the kinematic information of the microstructures, patterns, and configurations in complex fluids \cite{lin2012some}.
Similar to the macroscopic flow map $\x(\X, t)$, we can also introduce the microscopic flow map $\bm{q}(\X, \Qvec, t)$, where $\bm{X}$ and $\Qvec$ are Lagrangian coordinates in physical and configuration spaces respectively. The corresponding microscopic velocity ${\bf V}$ is defined as
\begin{equation} 
{\bf V}(\x(\X, t), \bm{q}(\X, \Qvec, t), t) = \frac{\dd}{\dd t}\bm{q}(\X, \Qvec, t).
\end{equation}
Due to the conservation of mass, the number density distribution function $\psi(\x, \qvec, t)$ satisfies the following kinematics
\begin{equation}\label{Kinemtic}
\pp_t \psi + \nabla_{\x} \cdot (\uvec \psi) + \nabla_{\qvec} \cdot ({\bm V} \psi) = 0,
\end{equation}
where $\uvec$ and ${\bm V}$ are effective velocities 
in the macroscopic domain and the microscopic configuration space respectively.
After the specification of the kinematics (\ref{Kinemtic}), the micro-macro system can be modeled through the energy-dissipation law
\begin{equation}\label{ED_MM}
  \begin{aligned}
   &  \frac{\dd}{\dd t} \int_{\Omega} \left[\frac{1}{2} \rho |\uvec|^2 + \lambda \int \psi (\ln \psi - 1) + \psi U \dd \qvec \right] \dd \x = - \int_{\Omega} \left[  \eta |\nabla \uvec|^2 + \int_{\mathbb{R}^d} \frac{\lambda}{\xi} \psi |{\bm V} - \nabla \uvec \qvec|^2 \dd \qvec \right] \dd \x, \\
  \end{aligned}
\end{equation}
where 
$\mathcal{K} = \int_{\Omega} \frac{1}{2} \rho |\uvec|^2 \dd \x$ is the kinetic energy, $\lambda$ is a constant that represents the ratio between the kinetic energy and the elastic energy, $U(\qvec)$ is the spring potential energy, and $\xi$ is the constant that is related to the polymer relaxation time. We assume that $\uvec$ satisfies the incompressible condition $\nabla \cdot \uvec = 0$. The second-term in the dissipation accounts for the relative friction of microscopic particle to the macroscopic flow, where $(\nabla \uvec) \qvec$ is velocity induced by the macroscopic flow due to the Cauchy-Born rule \cite{lin2007micro}. 
The Cauchy-Born rule states that the movement in configuration space follows the flow on the macroscopic level, i.e., $\qvec = {\sf F} \Qvec$ without the microscopic evolution, where $\Qvec$ are Lagrangian coordinates in the configuration space. A direct computation shows that
\begin{equation}
\widetilde{\bm V} = \frac{\dd}{\dd t} (\F \Qvec) = \nabla \uvec \F \Qvec = \nabla \uvec \qvec,
\end{equation}
which is the microscopic velocity induced by the macroscopic flow.

Now we are ready to perform the energetic variational approach in both microscopic and macroscopic scales. It is import to keep the ``separation of scales'' in mind when applying the LAP and the MDP in both scales. 
On the microscopic scale, 
since $\x(\X, t)$ is treated being independent from $\qvec(\Qvec, \X, t)$,
a standard energetic variational approach results in 
\begin{equation}\label{eq_V}
\frac{1}{\xi} \psi  ({\bm V} - \nabla \uvec \qvec) =  - \psi \nabla_{\qvec} (\ln \psi + U(\qvec)),
\end{equation}
where the right-hand side is obtained by the LAP, taking the variation of $-\mathcal{F}_{\qvec} =  - \int \psi \ln \psi + \psi U \dd \qvec $ with respect to ${\bm q}$, and the left-hand side is obtained by the MDP, taking the variation of $\mathcal{D}_{\qvec} = \frac{1}{2 \xi} \int \psi |{\bm V} - \nabla \uvec \qvec|^2 \dd \qvec$ with respect to ${\bm V}$ \cite{Giga2017, liu2020lagrangian}. 
Combining (\ref{eq_V}) with the kinematics (\ref{Kinemtic}), we have 
\begin{equation}
\psi_t + \nabla \cdot (\uvec \psi) + \nabla_{\qvec} \cdot (\nabla \uvec \qvec \psi) = \xi \left(\Delta_{\qvec} \psi + \nabla_{\qvec} \cdot (\psi \nabla_{\qvec} U) \right).
\end{equation}
On the macroscopic scale, 
due to the ``separation of scales'', we treat ${\bm q}(\X, \Qvec, t)$ as being independent from $\x(\X, t)$. The Cauchy-Born rule is taken into account by the dissipation term $|{\bm V} - (\nabla \uvec) \qvec|^2$. The action functional is defined by
\begin{equation}
  \begin{aligned}
\mathcal{A}(x) & = \int_{0}^T \int_{\Omega} \left[ \frac{1}{2} \rho |\uvec|^2 - \lambda \int_{\mathbb{R}^3} (\psi (\ln \psi - 1) + U \psi) \dd \qvec \right] \dd \x \dd t, \\
  \end{aligned}
\end{equation}
By the LAP, i.e., taking variation of $\mathcal{A}(\x)$ with respect to $\x$, we obtain
\begin{equation}
\frac{\delta\mathcal{A}}{\delta \x}= -\rho \x_{tt} = - \rho(\uvec_t + \uvec \cdot \nabla \uvec).
\end{equation}
Meanwhile, for the dissipation part, the MDP results in
\begin{equation}
\begin{aligned}
  \frac{\mathcal{\delta D}}{\delta \x_t} & =  - \eta \Delta \uvec + \frac{\lambda}{\xi} \nabla \cdot \int \psi (V - \nabla \uvec \qvec) \otimes \qvec \dd \qvec.  \\
      \end{aligned}
\end{equation}
Notice
\begin{equation}
  \begin{aligned}
    \frac{\lambda}{\xi} \nabla \cdot  \int \psi ({\bm V} - \nabla \uvec \qvec) \otimes \qvec \dd \qvec  & = \lambda \nabla \cdot \int \left( - \nabla_{\qvec} \psi \otimes \qvec  - \nabla_{\qvec} U \otimes \qvec \psi \right) \dd \qvec \\
& = - \lambda \nabla \cdot \left( \int \nabla_{\qvec} U \otimes \qvec \psi  \dd \qvec -  n  {\bf I} \right), \\
  \end{aligned}
\end{equation}
where the first equality is obtained by using (\ref{eq_V}).
Thus, the force balance condition leads to the macroscopic momentum equation:
\begin{equation}
\rho(\uvec_t + \uvec \cdot \nabla \uvec) + \nabla p= \eta \Delta \uvec +\nabla \cdot {\bm \tau},
\end{equation}
where $p$ is the Lagrangian multiplier for the incompressible condition $\nabla \cdot \uvec$, and ${\bm \tau}$ is the induced elastic stress tensor given by
\begin{equation}\label{tau_1}
{\bm \tau} = \lambda \left( \int_{\mathbb{R}^3} \psi \nabla_{\qvec} U \otimes \qvec \dd \qvec - n {\bf I}  \right).
\end{equation}
The form of the induced elastic stress tensor ${\bm \tau}$ is exactly the Kramers' expression of the polymeric stress \cite{li2007mathematical}, which reflects the microscopic contribution to the macroscopic flow. 

\begin{remark}
  In the above derivation, the induced elastic stress tensor is derived from the dissipation part of the energy-dissipation law. Alternatively, one can derive the equivalent induced elastic stress tensor from the conservative part \cite{lin2007micro, hyon2010energetic}. 
Due to the Cauchy-Born rule, we can assume that the configuration space follows the flow in the macroscopic scale, i.e., $\qvec = F \Qvec$ and ${\bf V} = \nabla \uvec \qvec$. Thus, the macroscopic action functional can be defined by
\begin{equation}
  \begin{aligned}
    \mathcal{A}(\x) = \int_{0}^T \int_{\Omega_0} \frac{1}{2} \rho_0 |\x_t|^2 - \lambda \int \psi_0  \left(\ln \psi_0 - 1 \right) + U(F \Qvec) \psi_0 \dd \Qvec \dd \X 
  \end{aligned}
\end{equation}
and the macroscopic dissipation is simply $\mathcal{D} = \frac{1}{2} \int \eta |\nabla \uvec|^2 \dd \x$ (the second term in the dissipation vanishes since the Cauchy-Born rule is used).
By taking variation of $\mathcal{A}(\x)$ with respect to $\x$, we have \cite{lin2007micro}
\begin{equation}
\frac{\delta \mathcal{A}}{\delta \x} =  - \rho \x_{tt} + \lambda \nabla \cdot (\int \psi \nabla_{\qvec} U \otimes \qvec \dd \qvec ).                                              
\end{equation}
Meanwhile, $\frac{\delta \mathcal{D}}{\delta \x_t} = - \eta \Delta \uvec.$
Hence, we end up with the same macroscopic equation with ${\bm \tau}$ given by
\begin{equation}
  {\bm \tau} = \lambda \int_{\mathbb{R}^3} \psi \nabla_{\qvec} U \otimes \qvec \dd \qvec,
\end{equation}
which is equivalent to the (\ref{tau_1}) in the incompressible case since $\nabla \cdot (-  n {\bf I})$ will contribute to the pressure and can be dropped \cite{li2007mathematical}.
\end{remark}

The classic energetic variational approach, as well as other variational principle \cite{beris2001bracket, doi2011onsager, grmela2014contact, grmela1997dynamics, ottinger1997dynamics}, which is indeed based on classical mechanics, cannot be applied to systems involving chemical reactions directly.
Since 1950’s, a large amount of works tried to developed a Onsager type variational theory for reaction kinetics by building analogies between Newtonian mechanics
and chemical reactions \cite{bataille1978nonequilibrium, beris1994thermodynamics, biot1977variational,grmela1993thermodynamics, grmela2012fluctuations, grmela2021multiscale, mielke2011gradient, oster1974chemical}. For instance,
a dissipation potential formulation of chemical reactions was introduced in \cite{grmela1993thermodynamics}. The formulation was extended to general mass-action kinetics involving inertia and fluctuations under the GENERIC framework in \cite{grmela2012fluctuations, pavelka2018multiscale, grmela2021multiscale}.
Motivated by these pioneering work,  the energetic variational formulation to chemical reactions was developed in a recent work \cite{wang2020field} by using the reaction trajectory ${\bm R}$ as the state variable.
The reaction trajectory, also known as the extent of reaction or degree of advancement, was originally introduced by De Donder \cite{de1927affinite, de1936thermodynamic}.
For a general reversible chemical reaction system containing $N$ species $\{ X_1, X_2, \ldots X_N \}$ and $M$ reactions, represented by
\begin{equation}
\ce{ $\alpha_{1}^{l} X_1 + \alpha_{2}^{l}X_2 + \ldots \alpha_{N}^{l} X_N$ <=> $\beta_{1}^{l} X_1 + \beta_{2}^{l}X_2 + \ldots \beta_{N}^{l} X_N$}, \quad l = 1, \ldots, M.
\end{equation}
We can define a reaction trajectory ${\bm R} \in \mathbb{R}^M$, where each component $R_l$ accounts for the “number” of $l$-th chemical reactions that has occurred in the forward direction by time $t$. The relation between species concentration ${\bm c} \in \mathbb{R}_{+}^N$ and the reaction trajectory ${\bm R}$ is given by
\begin{equation}\label{kinematics_R}
{\bm c} = {\bm c}_0 + {\bm \sigma} {\bm R},
\end{equation}
where ${\bm c}_0$ is the initial concentration, and $\bm{\sigma} \in \mathbb{R}^{N \times M}$ is the stoichiometric matrix with $\sigma_{il} = \beta^l_i - \alpha^l_i$. One can view (\ref{kinematics_R}) as the kinematics of the chemical reaction system \cite{liu2020structure}. With the kinematics (\ref{kinematics_R}), one can reformulate the free energy $\mathcal{F}$, which is a functional of ${\bm c}$, in terms of the reaction trajectory ${\bm R}$ \cite{wang2020field, liu2020structure}, Moreover, notice that
\begin{equation}
\frac{\delta \mathcal{F}}{\delta R_l} =  \sum_{i = 1}^N \sigma_{il}  \frac{\delta \mathcal{F}}{\delta c_i} = \sum_{i = 1}^N \sigma_{il} \mu_i,
\end{equation}
which is exactly the affinity of $l-$the chemical reaction, as defined by De Donder \cite{de1936thermodynamic}. It is worth pointing out that ${\bm R}$ corresponds to the internal state variable defined in \cite{coleman1967thermodynamics}.

The affinity plays a role of the ``force'' that drives the chemical reaction, which vanishes at the chemical equilibrium \cite{kondepudi2014modern}. The reaction trajectory ${\bf R}$ is the conjugate variable of the chemical affinity, which is analogous to the flow map $\x(\X, t)$ in mechanical systems \cite{oster1974chemical}. The reaction rate ${\bm r}$ is defined as $\pp_t {\bm R}$, which can be viewed as the reaction velocity \cite{kondepudi2014modern}. 
Similar to a mechanical system, the reaction rate can be obtained from a prescribed energy-dissipation law in terms of ${\bm R}$ and $\pp_t {\bm R}$:
\begin{equation}
  \frac{\dd}{\dd t} \mathcal{F}[{\bm R}] = - \mathcal{D}_{\rm chem}[{\bm R}, \pp_t {\bm R}],
  \end{equation}
  where $\mathcal{D}_{\rm chem}[{\bm R}, \pp_t {\bm R}]$ is the rate of energy dissipation due to  the chemical reaction procedure. 
  Since the linear response assumption for chemical system may be not valid unless at the last stage of chemical reactions \cite{beris1994thermodynamics, de2013non}, $\mathcal{D}_{\rm chem}$ is not quadratic in terms of $\pp_t {\bm R}$ in general.
  For a general nonlinear dissipation $$\mathcal{D}_{\rm chem}[{\bm R}, \pp_t {\bm R}] = \left( {\bm \Gamma}({\bm R}, \pp_t {\bm R}), \pp_t {\bm R}  \right) = \sum_{l=1}^M \Gamma_l ({\bm R}, \pp_t {\bm R}) \pp_t R_l \geq 0,$$
  the reaction rate can be derived as \cite{wang2020field, liu2020structure}:
  \begin{equation}\label{eq_Rl}
  \Gamma_l({\bm R}, \pp_t {\bm R}) = - \frac{\delta \mathcal{F}}{\delta R_l},
  \end{equation}
  which is the ``force balance'' equation for the chemical part \cite{wang2020field, liu2020structure}. It is often assumed that $\Gamma_l({\bm R}, \pp_t {\bm R}) = \Gamma_l (R_l, \pp_t R_l)$. So equation (\ref{eq_Rl}) specify the reaction rate of $l-$the chemical reaction. 
  In this formulation, the choice of the free energy determines the chemical equilibrium, while the choice of the dissipation functional $\mathcal{D}_{\rm chem}[{\bm R}, \pp_t {\bm R}]$ determines the reaction rate.
  


%


\subsection{Micro-macro model for wormlike micellar solutions}
Now we are ready to derive a thermodynamically consistent two-species micro-macro model for wormlike micellar solutions.
Following the setting of the VCM model \cite{vasquez2007network}, we consider there exist only two species in the system. A molecule of species $A$ 
can break into two molecules of species $B$, and two molecules of species $B$ can reform species $A$.
At a microscopic level, molecules of both species are modeled as elastic dumbbells as in classical models of dilute polymeric fluids \cite{bird1987dynamics, doi1988theory}. We denote the number density distribution of finding each molecule with end-to-end vector $\bm{q}$ at position $\bm{x}$ by $\psi_A(\bm{x}, \bm{q}, t)$ and $\psi_B(\bm{x},  \bm{q}, t)$ respectively.
The number density of species $\alpha$ is defined by
\begin{equation}
n_{\alpha} (\x, t) = \int \psi_{\alpha} \dd \bm{q}.
\end{equation}
We should emphasize that this is a coarse-grained description and the end-to-end vector $\qvec$ has no information on the length of polymer chains.

In general, the breakage and combination processes can be regarded as chemical reactions 
\begin{equation}\label{reaction}
\qvec + \qvec' \ce{<=>} \qvec'',
\end{equation}
where $\qvec$ and $\qvec'$ are end-to-end vectors of species $B$ and $\qvec''$ is an end-to-end of species $A$ [see Fig. \ref{Fig0}(a) for illustration]. We denote the forward and backward reaction rate of (\ref{reaction}) by $W^+(\qvec, \qvec'; \qvec'')$ and $W^-(\qvec, \qvec'; \qvec'')$ respectively. The kinematics of $\psi_A$ and $\psi_B$ can be written as
\begin{equation}
  \begin{cases}
    & \pp_t \psi_A + \nabla \cdot (\uvec_A \psi_A) + \nabla_{\qvec} \cdot ({\bm V}_A \psi_A) = \int R_t(\qvec', \qvec''; \qvec ) \dd \qvec' \dd \qvec'' \\
    & \pp_t \psi_B +  \nabla  \cdot (\uvec_B  \psi_B) + \nabla_{\qvec} \cdot ({\bm V}_B \psi_B) = - \int R_t(\qvec, \qvec'; \qvec'' ) \dd \qvec' \dd \qvec'' - \int R_t(\qvec', \qvec; \qvec'') \dd \qvec' \dd \qvec'' \\
    & R_t(\qvec, \qvec'; \qvec'') = W^{+}(\qvec, \qvec'; \qvec'') \psi_B(\qvec)\psi_B(\qvec')  - W^{-}(\qvec, \qvec'; \qvec'') \psi_A (\qvec''),  \\
  \end{cases}
  \end{equation}
where $\uvec_{\alpha}$ and ${\bm V}_{\alpha}$ are effective macroscopic and microscopic velocities.
Different models can be obtained by choosing $W^+(\qvec, \qvec'; \qvec'')$ and $W^-(\qvec, \qvec'; \qvec'')$ differently.
In this paper, we 
take
\begin{equation}\label{Assumption_reaction}
  W^{\pm}(\qvec, \qvec'; \qvec'') \neq 0 \quad \text{if and only if} \quad \qvec = \qvec' = \qvec'',
\end{equation}
which corresponds to the case that an $A$ molecule at position $\x$ with end-to-end vector $\qvec$ can only break into two $B$ molecules with same end-to-end vector, and the combination process can only happen between two $B$ molecules at the same position $\x$ with the same end-to-end vector. This is a special case of a reversible microscopic reaction mechanism $\qvec + \qvec \ce{<=>} \alpha \qvec$, illustrated in Fig. \ref{Fig0}(b), with $\alpha = 1$. $\alpha$ can be viewed as a parameter for fast conformational changes of species $A$.
Within this assumption, one can have a {\emph detailed balance} condition for each $\x$ and $\qvec$ and the kinematics can reduce to 
\begin{equation}\label{kinematic}
  \begin{cases}
    & \pp_t \psi_A + \nabla \cdot (\uvec_A \psi_A) + \nabla_{\qvec} \cdot (V_A \psi_A) = - R_t\\  
    & \pp_t \psi_B +  \nabla  \cdot (\uvec_B  \psi_B) + \nabla_{\qvec} \cdot (V_B \psi_B) = 2 R_t, \\
  \end{cases}
\end{equation}
where $R(\x, \qvec, t)$ is the reaction trajectory for the breakage and combination for given $\qvec$ and $\x$. We should emphasize that the assumption here is only for the mathematical simplicity, which may not fully reflect the complicated physical scenario. 


\begin{figure}[!h]
  \centering
  \begin{overpic}[width = 0.45 \textwidth]{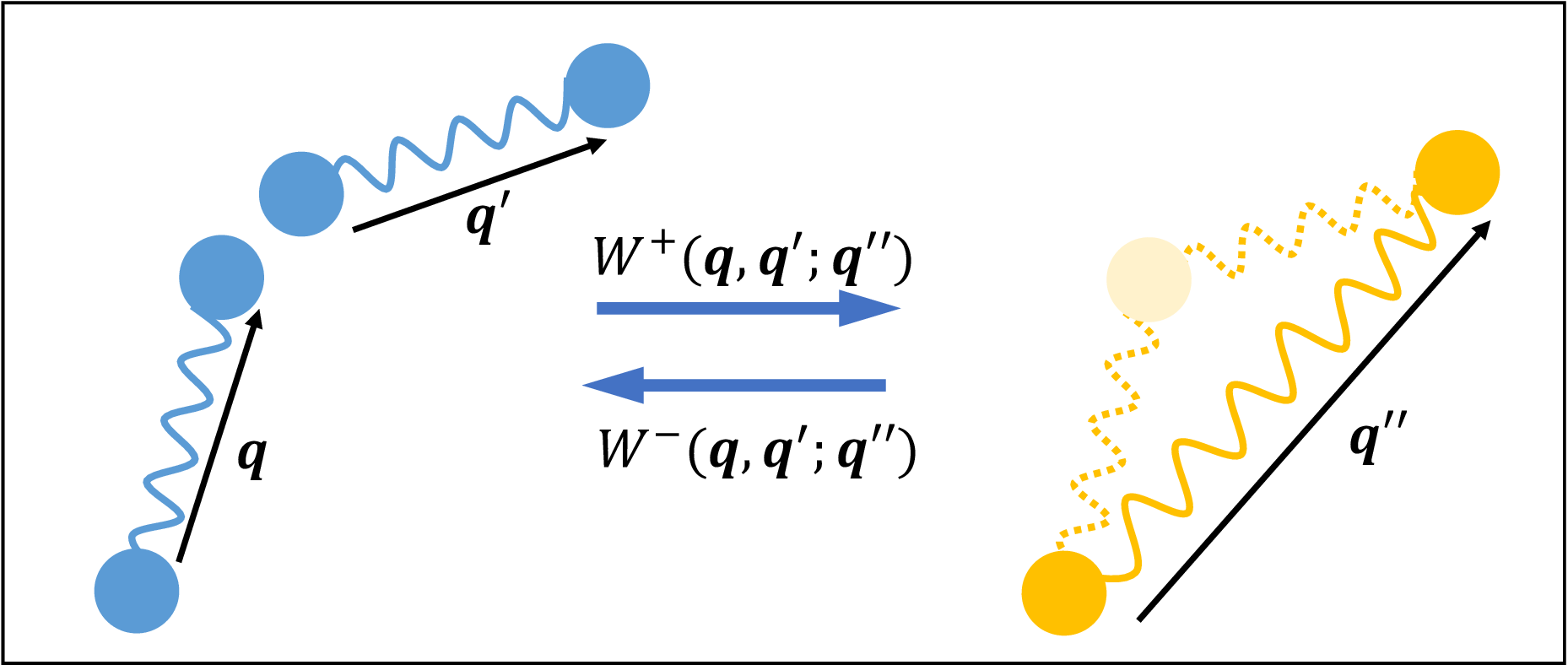}
  \put(0, 37){ {(a)} }
  \end{overpic}
  \hspace{1em}
  \begin{overpic}[width = 0.45 \textwidth]{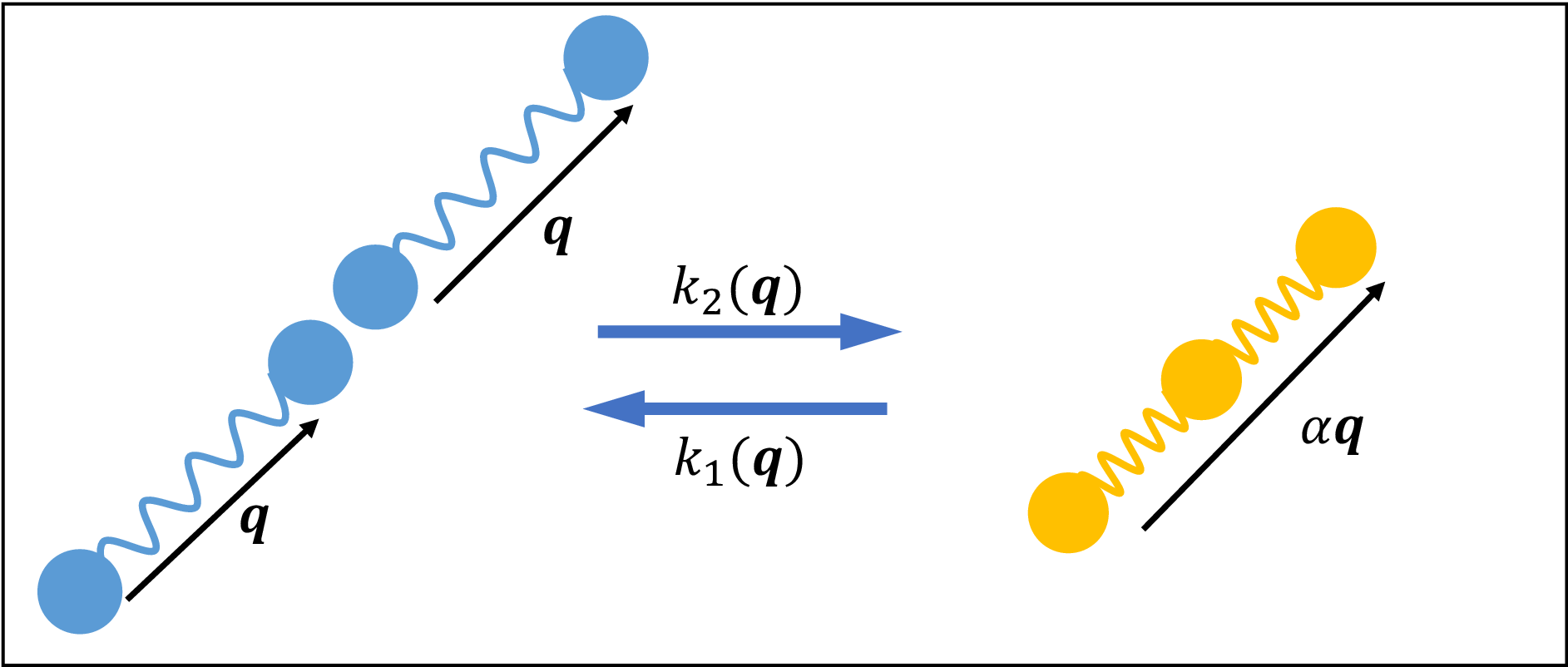}
    \put(0, 37){ {(b)} }
  \end{overpic}
  \caption{Schematic diagram of breakage and combination processes in wormlike micellar solutions, in which different species are indicated by different colors. The reaction mechanism considered in this paper is (b) with $\alpha = 1$.}\label{Fig0}
\end{figure}

  \begin{remark}
  In the original VCM model \cite{vasquez2007network}, the authors assume that 
    \begin{equation}\label{kinematic_VCM}
      \begin{cases}
        & \pp_t \psi_A + \nabla \cdot (\uvec_A \psi_A) + \nabla_{\qvec} \cdot ({\bm V}_A \psi_A) = -  k_1 \psi_A + k_2 \psi_B * \psi_B\\  
        & \pp_t \psi_B +  \nabla  \cdot (\uvec_B  \psi_B) + \nabla_{\qvec} \cdot ({\bm V}_B \psi_B) =  2 k_1 \psi_A - 2 k_2 \psi_B * \psi_B, \\
      \end{cases}
    \end{equation}
    where
  \begin{equation}
  \psi_B * \psi_B  = \int \psi_B(\x, \qvec - \tilde{\qvec}, t) \psi_B(\x, \tilde{\qvec}, t) \dd \tilde{q}
  \end{equation}
  The advantage of the assumption (\ref{kinematic_VCM}) is that the system will satisfies the law of mass action for $n_A$ and $n_B$ in the macroscopic scale, that is
    \begin{equation}
      \begin{cases}
      & \pp_t n_A + \nabla \cdot (n_A \uvec_A) = - k_1 n_A + k_2 n_B^2 \\
      & \pp_t n_B + \nabla \cdot (n_B \uvec_B) = 2 k_1 n_A - 2 k_2 n_B^2. \\
      \end{cases}
    \end{equation}
  by integrating both sides of (\ref{kinematic_VCM}) with respect to $\qvec$. However, as pointed out in \cite{adams2018nonlinear}, the reaction mechanism in VCM model is not microscopically reversible, as a $A$ molecule can only break in the middle to give two equal-length $B$ molecule and two $B$ macromolecules can combine through adding the end-to-end vector. As a consequence, it seems to difficult to obtain a variational structure for the breakage and combination mechanism (\ref{kinematic_VCM}). To repair the thermodynamic problem, in \cite{adams2018nonlinear}, the author proposed a microscopic reversible reaction mechanism with
  \begin{equation}
W^{\pm}(\qvec, \qvec'; \qvec'') \neq 0, ~ \text{if and only if} ~  \qvec'' = \qvec + \qvec'~\text{or}~\qvec'' = \qvec - \qvec',
  \end{equation}
  in their Brownian dynamics simulations.
  \end{remark}

  \begin{remark}
    The previous kinematic assumption for the breakage and combination process is based on a two-species approach. An alternative approach, which is a direct extension of Cates' original work, is to view all micelles as one species with different end-to-end vector $\qvec$. Then the reaction assumption \eqref{reaction} gives a kinematics
    \begin{equation}
      \begin{cases}
        & \pp_t \psi + \nabla \cdot (\uvec \psi) + \nabla_{\qvec} \cdot ({\bm V} \psi) \\
        & \quad  = - \int R_t(\qvec, \qvec'; \qvec'' ) \dd \qvec' \dd \qvec'' - \int R_t(\qvec', \qvec; \qvec'') \dd \qvec' \dd \qvec'' + \int R_t(\qvec', \qvec''; \qvec) \dd \qvec' \dd \qvec'', \\
        &  R_t(\qvec, \qvec'; \qvec'') = W^{+}(\qvec, \qvec'; \qvec'') \psi(\qvec)\psi(\qvec')  - W^{-}(\qvec, \qvec'; \qvec'') \psi (\qvec'')  \\
      \end{cases}
    \end{equation}
  
  \end{remark}

Similar to the one-species micro-macro model, the total energy of the system can be written as 
\begin{equation}
  \begin{aligned}
    E^{total}  =  \int_{\Omega} \Big[ \frac{1}{2} \rho |\uvec|^2 + \lambda \int  & \psi_A ( \ln \psi_A - 1) +  \psi_A U_A(\qvec)  + \psi_B  (\ln \psi_B - 1) + \psi_B U_B(\qvec)  \dd \qvec \Big] \dd \x, \\
  \end{aligned}
\end{equation}
where $\uvec$ is the velocity field of the macroscopic flow satisfying the incompressible condition $\nabla \cdot \uvec = 0$, $\lambda$ is the ratio between the macroscopic kinetic energy and microscopic elastic energy, and $U_{\alpha}(\qvec)$ is the potential energy associated with each species.

Throughout this paper, 
we disregard the diffusive effects of $A$ and $B$, and assume $\uvec_A = \uvec_B = \uvec$, which is the velocity of the macroscopic fluids.
Then the dissipation 
can be formulated as
\begin{equation}
  \begin{aligned}
  \triangle =  - \int_{\Omega} \Bigl[ & \eta | \nabla \uvec |^2  + \lambda  \int  \frac{\psi_A}{\xi_A} |\bm{V}_A - \nabla \uvec  \qvec|^2 +  \frac{\psi_B}{\xi_B} |\bm{V}_B - \nabla \uvec \qvec|^2 + \pp_t R ~  \Gamma(R, \pp_t R) \dd \qvec 
             \Bigr] \dd \x,  \\ 
  \end{aligned}
\end{equation}
where  $\xi_{\alpha}$ is a constant related to the relaxation time of each species, and $\pp_t R ~ \Gamma(R, \pp_t R) \geq 0$ is the additional dissipation due to the breakage and combination process. 
Different choices of $\Gamma(R, \pp_t R)$ determine different reaction rates. 
A typical choice of $\Gamma(R, \pp_t R)$ is
\begin{equation}
\Gamma(R, \pp_t R) = \ln \left(\frac{\pp_t R}{\eta(c(R))} + 1 \right).
\end{equation}
Recall (\ref{eq_Rl}), we can obtain the equation of $R$ as
\begin{equation}
  \ln \left(\frac{\pp_t R}{\eta(c(R))} + 1 \right) = \mu_A - 2 \mu_B,
\end{equation}
where $\mu_{\alpha} = \frac{\delta \mathcal{F}}{\delta \psi_{\alpha}} = \ln \psi_{\alpha} + U_{\alpha}$ is the chemical potential of species $A$ and $B$.
A further calculation leads to 
\begin{equation}\label{reaction_rate_1}
\pp_t R = \eta(R) \left( \exp\left( - (U_B - 2 U_A)  \right) \frac{\psi_A}{\psi_B^2} - 1 \right).
\end{equation}
If $\eta(R) = k_2 (\qvec) \psi_B^2$, (\ref{reaction_rate_1}) can be further simplified as
\begin{equation}\label{micro_LMA}
\pp_t R = k_1(\qvec) \psi_A - k_2(\qvec) \psi_B^2,
\end{equation}
where $$k_1(\qvec) = \frac{k_2(\qvec)}{K_{eq}(\qvec)}, \quad
  K_{eq}(\qvec) =  \frac{\psi_A^{\infty}}{ (\psi_B^{\infty})^2} = \exp \left(2 U_B - U_A \right), $$ which is the law of mass action at the microscopic level. $K_{eq} (\qvec)$ is the equilibrium constant for given $\qvec$.

The derivation of the mechanical part of the two-species model is almost same to that in the one-species case. In the microscopic scale, a standard EnVarA leads to 
\begin{equation}
 \psi_{\alpha} \nabla_{\qvec} \mu_{\alpha} = - \dfrac{1}{\xi_{\alpha}} \psi_{\alpha} (V_{\alpha} - (\nabla \uvec) \qvec) 
\end{equation}
that is
\begin{equation}
  \psi_{\alpha} V_{\alpha}  = - \xi_{\alpha}  (\psi_{\alpha} \nabla_{\qvec} U_{\alpha} + \nabla_{\qvec} \psi_{\alpha}) + (\nabla \uvec) \qvec) \psi_{\alpha}.
\end{equation}
Hence, the microscopic equation is given by
\begin{equation}\label{eq_psi}
  \begin{cases}
   & \pp_t \psi_{A} + \uvec \cdot \nabla \psi_{A} + \nabla_{\qvec} \cdot (\nabla \uvec \qvec \psi_{A}) - \xi_{A} \nabla_{\qvec} \cdot (\nabla_{\qvec} \psi_{A} + \nabla_{\qvec} U_{A}\psi_{A}) = - \pp_t R, \\
  &   \pp_t \psi_{B} + \uvec \cdot \nabla \psi_{B} + \nabla_{\qvec} \cdot (\nabla \uvec \qvec \psi_{B}) -  \xi_{B} \nabla_{\qvec} \cdot (\nabla_{\qvec} \psi_{B} + \nabla_{\qvec} U_{B}\psi_{B} ) = 2 \pp_t R, \\
  \end{cases}
\end{equation}
where $\pp_t R$ is defined in (\ref{micro_LMA}). On the macroscopic scale, similar to the one species case, by an energetic variational approach, we can obtain
\begin{equation}\label{eq_moment}
  \rho (\uvec_t + (\uvec \cdot \nabla) \uvec) + \nabla p = \eta \Delta \uvec +  \lambda \nabla \cdot{\bm \tau}
\end{equation}
where $\bm{\tau}$ is the induced stress from the microscopic configurations
\begin{equation}\label{def_tau_MM}
  \begin{aligned}
\bm{\tau} & =  \int \psi_A \nabla_{q} \mu_A \otimes q \dd q + \int \psi_B \nabla_{q} \mu_B \otimes \qvec \dd \qvec  \\
& = \int  \left( \nabla_{q} U_A \otimes \qvec \psi_A + \nabla_{q} U_B \otimes \qvec \psi_B \right) \dd \qvec  - (n_A + n_B) \mathrm{I} \\
  \end{aligned}
\end{equation}

Hence, the final macro-micro system is given by
\begin{equation}\label{Final_MM}
  \begin{cases}
    &  \rho (\pp_t \uvec + (\uvec \cdot \nabla) \uvec) + \nabla p = \eta \Delta \uvec + \lambda \nabla \cdot \bm{\tau}   \\
    & \nabla \cdot \uvec = 0 \\
    &  \pp_t \psi_{A} + \uvec \cdot \nabla \psi_{A} + \nabla_{\qvec} \cdot (\nabla \uvec \qvec \psi_{A}) - \xi_{A} \nabla_{\qvec} \cdot (\nabla_{\qvec} \psi_{A} + \nabla_{\qvec} U_A \psi_A  ) =  - \pp_t R  \\
    &  \pp_t \psi_{B} + \uvec \cdot \nabla \psi_{B} + \nabla_{\qvec} \cdot (\nabla \uvec \qvec \psi_{B}) - \xi_B \nabla_{\qvec} \cdot (\nabla_{\qvec} \psi_{B} + \nabla_{\qvec} U_B \psi_B ) = 2 \pp_t R  \\  
  \end{cases}
\end{equation}
where
\begin{equation}
  \pp_t R = k_1 (\qvec) \psi_A - k_2 (\qvec) \psi_B^2,
  \end{equation}
  and ${\bm \tau}$ is the stress tensor given by (\ref{def_tau_MM}).
According to the previous derivation,  it is easy to show that the system satisfies the following energy-dissipation property:
\begin{equation}\label{Energy_Id}
  \begin{aligned}
& \frac{\dd}{\dd t} \int \left[ \frac{1}{2} \rho |\uvec|^2 + \lambda \int \psi_A (\ln \psi_A - 1 + U_A) + \psi_B (\ln \psi_B - 1 + U_B) \dd \qvec \right]  \dd \x \\
& = - \int \left[  \eta |\nabla \uvec|^2 +  \frac{\lambda}{\xi_A} \int \psi_A |\nabla_{\qvec} (\ln \phi_A + U_A)|^2 \dd \qvec + \frac{\lambda}{\xi_B} \int \psi_B |\nabla_{\qvec} (\ln \psi_B + U_B) |^2 \dd \qvec \right. \\
&  \left. \quad \quad \quad  + \lambda \int (k_1 (\qvec) \psi_A - k_2 (\qvec) \psi_B^2) \ln \left( \frac{k_1 (\qvec) \psi_A}{k_2 (\qvec) \psi_B^2}  \right) \dd \qvec \right] \dd \x,  \\ 
  \end{aligned}
\end{equation}


\section{Moment closure models}
The micro-macro model (\ref{Final_MM}) provides a thermodynamically consistent multi-scale description to wormlike micellar solutions. However, it might be difficult to study this model directly, as the microscopic equation (\ref{eq_psi}) is high dimensional.  Notice that the macroscopic stress tensor only involves the zeroth and second moments of the number distribution functions of two species, it is a natural idea to derive a coarse-grained macroscopic equation from the original micro-macro model through moment closure.
Moment closure is a powerful tool to obtain coarse-grained macroscopic constitutive equations from more detailed micro-macro models for complex fluids \cite{doi1981molecular,du2005fene, feng1998closure, ottinger2009stupendous,wang1997comparative}.
One challenge in moment closure is to preserve the thermodynamic structures, i.e., the coarse-grained system should satisfy a energy-dissipation law analogous to the energy-dissipation law of the original system \cite{ottinger2009stupendous, hyon2008maximum}. 
The presence of the chemical reaction imposes additional difficulties for closure approximations. 

Throughout this section, we assume the potential energy $U_{\alpha}$ to be
\begin{equation}\label{UA_UB}
U_A = \frac{1}{2} H_A |\qvec|^2 + \sigma_A, \quad U_B = \frac{1}{2} H_{B} |\qvec|^2 + \sigma_B,
\end{equation}
where $\sigma_A$ and $\sigma_B$ are constants related to the equilibrium of the breakage and combination procedure, $H_A$ and $H_B$ are Hookean spring constants associated with species $A$ and $B$. Moreover, we assume that
\begin{equation}
  H_A = 2 H_B,
  \end{equation}
  then $K_{eq} = \exp(2 \sigma_B - \sigma_A)$ is a constant, which enables us to have a model with both $k_1$ and $k_2$ being constants. Same assumption is used in the GCB model \cite{germann2013nonequilibrium}. 
Other types of potential energies can be considered but will result in more complicated closure systems. 

\begin{remark}
  The assumption $H_A = 2 H_B$ is the consequence of the detailed balance condition 
  for the reaction $\qvec + \qvec \ce{<=>[k_2][k_1]} \qvec$ with constant reaction rates $k_i$, i.e.,
  \begin{equation}
k_1 \psi_A^{\infty}(\qvec) = k_2 (\psi_B^{\infty} (\qvec) )^2, \quad \psi_{\alpha}^{\infty} = C_{\alpha} \exp\left( - \frac{1}{2} H_{\alpha} \qvec^{\rm T} \qvec \right),
  \end{equation}
  where $\psi_{\alpha}^{\infty}$ is an equilibrium number density distribution for each species and  $C_{\alpha}$ is a constant
  . If the reaction mechanism $\qvec + \qvec \ce{<=>[k_2][k_1]} \alpha \qvec$ is assumed, then the detailed balance condition requires 
  \begin{equation}
    \alpha^2 H_A = 2 H_B.
  \end{equation}
  We have $2 H_A = H_B$ for $\alpha = 2$, which is assumption in the VCM model.

  \end{remark}

With the assumption $H_A = 2 H_B$, the global equilibrium distribution of the system is given by
\begin{equation}
\psi^{\infty}_A = \frac{n_A^{\infty}}{ (2 \pi H_A^{-1})^{d/2}} \exp \left(- \frac{1}{2} H_A \qvec^{\rm T} \qvec \right), \quad \psi^{\infty}_B = \frac{n_B^{\infty}}{ (2 \pi H_B^{-1})^{d/2}} \exp\left(- \frac{1}{2} H_B \qvec^{\rm T} \qvec\right),
\end{equation} 
where
$n_{A}^{\infty}$ and $n_{B}^{\infty}$ are number densities at the global equilibrium. 
Correspondingly, the second moments at the global equilibrium are given by
\begin{equation}
 {\bf A}_{\rm eq} = \int \qvec \otimes \qvec \psi^{\infty}_A \dd \x = \frac{n_A^{\infty}}{H_A} {\bf I}, \quad {\bf B}_{\rm eq} = \int \qvec \otimes \qvec \psi^{\infty}_B \dd \x = \frac{n_B^{\infty}}{H_B} {\bf I}.
   \end{equation}
Let $K_{eq}^{\rm macro} = \frac{n_A^{\infty}}{ (n_B^{\infty})^2}$ be the macroscopic equilibrium constant and a direct computation shows that
\begin{equation}\label{Micro_Macro_K_eq}
K_{\rm eq} = e^{2 \sigma_B - \sigma_A}  = \frac{\psi_A^{\infty}}{(\psi_{B}^{\infty})^2} 
= \frac{2^d \pi^{d/2}}{H_B^{d/2}}   K_{eq}^{\rm macro},
\end{equation}
which reveals the connection between $K_{eq}$ and $K_{eq}^{\rm macro}$.




\subsection{Maximum entropy closures}
Maximum entropy closures, also known as quasi-equilibrium approximations \cite{gorban2001corrections, hyon2008maximum, klika2019dynamic, pavelka2018multiscale, wang2008crucial}, have been successfully used to derive effective macroscopic equations from the micro-macro multi-scale models for polymeric fluids, including nonlinear dumbbell models \cite{wang2008crucial, hyon2008maximum} and liquid
crystal polymers \cite{ball2010nematic, ilg2003canonical, han2015microscopic,yu2010nonhomogeneous}. For nonlinear dumbbell models with FENE potential, it has been shown that maximum entropy closure can capture the hysteretic behavior and maintain the energy-dissipation property \cite{hyon2008maximum, wang2008crucial}.

The idea of the maximum entropy closure is to maximize the ``relative entropy'' subjected to moments \cite{gorban2001corrections, hyon2008maximum, ottinger2009stupendous, wang2008crucial}. For our system, we can approximate $\psi_{\alpha}$ ($\alpha = A, B$) based on its zeroth moment $n_{\alpha}$ and second moment ${\bf M}_{\alpha}$ by solving the constrained optimization problem  
\begin{equation}\label{MinProblem}
  \begin{aligned}
\psi_{\alpha}^{*} = \text{argmin}_{\mathcal{A}} \int_{\mathbb{R}^d} \psi \ln \psi + \psi U_{\alpha} (\qvec) \dd \qvec, 
  \end{aligned}
\end{equation}
where 
\begin{equation}
\mathcal{A} = \left\{  \psi: \mathbb{R}^d \rightarrow \mathbb{R}, \psi \geq 0 ~|~ \int \psi \dd \qvec = n_{\alpha}, \quad \int(\qvec \otimes \qvec)  \psi  \dd \qvec = {\bf M}_{\alpha}  \right\}.
\end{equation}

\begin{proposition}
For the Hookean potential $U_{\alpha} = \frac{1}{2} H_{\alpha} |\qvec|^2 + \sigma_{\alpha}$, the minimization problem (\ref{MinProblem}) has a unique minimizer $\psi^{*}_{\alpha}$ in the class $\mathcal{A}$ for given $n_{\alpha} > 0$ and a symmetric positive-definite matrix ${\bf M}_{\alpha}$. Moreover, $\psi^*_{\alpha}$ is given by
\begin{equation*}
  \psi_{\alpha}^* (\qvec) = \frac{n_{\alpha}}{(2 \pi)^{d/2} (\det \widetilde{\bf M}_{\alpha})^{1/2}} \exp(-\frac{1}{2} \qvec^{\rm T} \widetilde{\bf M}_{\alpha}^{-1} \qvec), 
  \end{equation*}
where $\widetilde{\bf M}_{\alpha} = {\bf M}_{\alpha} / n_{\alpha}$. We call $\psi_{\alpha}^*$ is the quasi-equilibrium state associated with $n_\alpha$ and ${\bf M}_{\alpha}$.
\end{proposition}

\begin{proof}

The solution to the constrained optimization problem (\ref{MinProblem}) is given by
\begin{equation}\label{Lg_var}
  \begin{aligned}
\frac{\delta}{\delta \psi} \Bigl\{  & \int \psi \ln \psi + U_{\alpha} (\qvec) \psi \dd \qvec + \lambda_0 \left[ \int \psi \dd \qvec - n_{\alpha} \right] + \sum_{ij} \lambda_{ij} \left[ \int q_i q_j \psi \dd \qvec - ({\bf M}_{\alpha})_{ij} \right] \Bigr\} = 0, \\
  \end{aligned}
\end{equation}
where $\lambda_0$ and $\lambda_{ij}$ are Lagrangian multipliers. From (\ref{Lg_var}), one can obtain that
\begin{equation}
\psi_{\alpha}^{*} = C \exp(- \frac{1}{2} H_{\alpha} |\qvec|^2 - \sigma_{\alpha})\exp \left(- \lambda_0- \sum_{ij} \lambda_{ij} q_i q_j \right),
\end{equation}
where $C > 0$ is a constant. Since $\int \psi^{*}_{\alpha} \dd \x = n_{\alpha}$, $\psi^{*}_{\alpha}$ can be written as
\begin{equation}
\psi^{*}_{\alpha} = \frac{n_{\alpha}}{Z(\lambda_{ij})}  \exp \left( - \frac{1}{2} H_{\alpha} |\qvec|^2- \sum_{ij} \lambda_{ij} q_i q_j \right) .
\end{equation}
where $Z(\lambda_{ij}) = \int \exp(- \frac{1}{2} H_{\alpha} |\qvec|^2) \exp \left(- \sum_{ij} \lambda_{ij} q_i q_j \right) \dd \qvec$ is the normalizing constant. 
Since $\psi_{\alpha}^{*}/n_{\alpha}$ is the multivariate normal distribution $\mathcal{N}(0, \bm{\Sigma})$ with the covariance matrix given by
\begin{equation}
 {\bm \Sigma} = (H_{\alpha} {\bf I} + 2 {\bm \Lambda})^{-1},
\end{equation}
which is uniquely determined by its second moment, i.e., ${\bm \Sigma} = (H_{\alpha} {\bf I} + 2 {\bm \Lambda})^{-1} = {\bf M}_{\alpha} / n_{\alpha}$ \cite{hu2007new}.  
\end{proof}

Thus, for given $n_A > 0$ , $n_B > 0$, positive-definite matrices ${\bf A}$ and ${\bf B}$, we can define the unique quasi-equilibrium states
\begin{equation}\label{closure_AB}
  \begin{aligned}
    & \psi^{*}_A = \frac{n_A}{(2 \pi)^{d/2} (\det \widetilde{\bf A})^{1/2}} \exp \left( - \frac{1}{2} \qvec^{\rm T} \widetilde{\bf A}^{-1} \qvec \right).  \\
  & \psi^{*}_B  = \frac{n_B}{(2 \pi)^{d/2} (\det \widetilde{\bf B})^{1/2}} \exp \left( - \frac{1}{2} \qvec^{\rm T} \widetilde{\bf B}^{-1} \qvec \right), \\
  \end{aligned}
\end{equation}
where $\widetilde{\bf A} = {\bf A} / n_A$ and $\widetilde{\bf B} = {\bf B} / n_B$ are conformation tensors \cite{germann2013nonequilibrium}. 
We call the manifold formed by all quasi-equilibrium distributions as the \emph{quasi-equilibrium manifold}, denoted by
\begin{equation}\label{QE_M}
\mathcal{M}^* = \left\{ \psi^* = \frac{n}{(2 \pi)^{d/2} (\det \widetilde{\bf M})^{1/2}} \exp \left( - \frac{1}{2} \qvec^{\rm T} \widetilde{\bf M}^{-1} \qvec \right) ~|~ n > 0, \widetilde{\bf M}\text{~symmetric, positive-definite} \right \}
\end{equation}
For $\psi_{\alpha}^* \in \mathcal{M}^*$, its second moment ${\bf M}_{\alpha} = n_{\alpha} \widetilde{\bf M}_{\alpha}$ depends on its zeroth moment $n_{\alpha}$.

 \subsection{The moment closure model: variation-then-closure}
 We can apply the maximum entropy closure to the micro-macro model (\ref{Final_MM}) directly. 
Since $k_1$ and $k_2$ are constants, by integrating (\ref{kinematic}) over $\bm{q}$, we have
\begin{equation}\label{int_n}
  \begin{cases}
    & \pp_t n_A +  \nabla \cdot (n_A \uvec)  = - k_1 n_A + k_2 \int \psi_B^2 \dd \qvec \\
    & \pp_t n_B +  \nabla \cdot (n_B \uvec) = 2 k_1 n_A - 2 k_2 \int \psi_B^2 \dd \qvec. \\
  \end{cases}
\end{equation}
Meanwhile, multiplying both side of (\ref{kinematic}) by $\bm{q} \otimes \bm{q}$ and integrating over $\bm{q}$  arrives at
\begin{equation}\label{eq_M}
  \begin{cases}
    & \pp_t {\bf A} + (\uvec \cdot \nabla) {\bf A} 
    - (\nabla \uvec) {\bf A} - {\bf A} (\nabla \uvec)^{\rm T} 
    = \xi_A (2 n_A {\bf I} - 2 H_A {\bf A}) 
     - k_1 {\bf A} + k_2 \int \qvec \otimes \qvec \psi_B^2 \dd \qvec \\
    & \pp_t {\bf B} + (\uvec \cdot \nabla) {\bf B}  
    - (\nabla \uvec) {\bf B} - {\bf B} (\nabla \uvec)^{\rm T} 
    =  \xi_B (2 n_B {\bf I} - 2 H_B {\bf B}) 
     + 2 k_1 {\bf A} - 2 k_2 \int \qvec \otimes \qvec \psi_B^2 \dd \qvec.
  \end{cases}
\end{equation}
Therefore, for Hookean spring potentials and constant reaction rates, the moment closure is needed only due to the nonlinear reaction term in the microscopic scale. 
 With the maximum entropy closure (\ref{closure_AB}), these two terms can be computed out explicitly.
Indeed, notice that 
\begin{equation}
  \int (\psi_B^*)^2 \dd \qvec = \int_{\mathbb{R}^d} \frac{n_B^2}{Z_B^2} \exp (- \qvec^{\rm T} \widetilde{\bf B}^{-1} \qvec) \dd \qvec,
\end{equation}
by letting $\bm{q} = \frac{1}{\sqrt{2}} \widetilde{\qvec}$, we have
\begin{equation}
  \begin{aligned}
  \int (\psi_B^*)^2 \dd \qvec & =  2^{-d/2} \int_{\mathbb{R}^d} \frac{n_B^2}{Z_B^2} \exp (- \frac{1}{2} \widetilde{\qvec}^{\rm T} \widetilde{\bf B}^{-1} \widetilde{\qvec}) \dd \widetilde{\qvec}  = \frac{n_B^{d/2}}{ 2^d  \pi^{d/2}  (\det B)^{1/2}} n_B^2.
  \end{aligned}
\end{equation}
Hence,
\begin{equation}\label{close_n}
  \int k_1 \psi_A^* - k_2 (\psi_B^*)^2 \dd \qvec = k_1 n_A - \widetilde{k}_2 ({\bf B}) n_B^2,
 \end{equation}
where $\widetilde{k}_2 ({\bf B})$ is given by
\begin{equation}
\widetilde{k}_2 ({\bf B}) = \frac{n_{B}^{d/2}}{  2^d (\pi)^{d/2} (\det ({\bf B})  )^{1/2} }  k_2.
\end{equation}
Interestingly, in this case, the maximum entropy closure gives us the law of mass action on number densities $\ce{A <=>[k_1][\widetilde{k}_2] 2 B}$,
as in the VCM and GCB models. 
By a similar calculation, we have
\begin{equation}\label{close_M}
  \begin{aligned}
 k_2 \int (\psi_B^*)^2 (\qvec \otimes \qvec) \dd \qvec & = \frac{(\sqrt{2})^{-d}}{2} \int_{\mathbb{R}^d}  \frac{n_B^2}{Z_B^2} \exp ( - \frac{1}{2} \widetilde{\qvec}^{\rm T} \widetilde{\bf B}^{-1} \widetilde{\qvec} ) \tilde{\qvec} \otimes \tilde{\qvec} \dd \tilde{\qvec} = \frac{1}{2} \widetilde{k}_2 ({\bf B}) n_B {\bf B}. \\
  \end{aligned}
\end{equation}
Therefore, applying the maximum entropy approximation to (\ref{Final_MM}), we can obtain a moment closure system
\begin{equation}\label{Direct-closure}
  \begin{cases}
    &  \rho (\pp_t \uvec + (\uvec \cdot \nabla) \uvec) + \nabla p = \eta \Delta \uvec + \lambda \nabla \cdot \left(  H_A {\bf A} + H_B {\bf B} - (n_A + n_B) {\bf I} \right)   \\
    & \nabla \cdot \uvec = 0 \\
    & \pp_t n_{A} + \nabla \cdot (n_A \uvec ) = - k_1 n_A + \widetilde{k}_2 ({\bf B}) n_B^2, \\
    & \pp_t n_{B} + \nabla \cdot (n_B \uvec ) = 2 k_1 n_A - 2 \widetilde{k}_2 ({\bf B}) n_B^2, \\
    &\pp_t {\bf A} + (\uvec\cdot\nabla) {\bf A} - (\nabla \uvec) {\bf A} - {\bf A} (\nabla \uvec)^{\rm T} = 2 \xi_A (n_A {\bf I} - H_A {\bf A}) - k_1 {\bf A}+ \frac{1}{2} \widetilde{k}_2 ({\bf B}) n_B {\bf B} \\
    &\pp_t {\bf B} + (\uvec\cdot\nabla) {\bf B} - (\nabla \uvec) {\bf B} - {\bf B} (\nabla \uvec)^{\rm T} =  2 \xi_B (n_B {\bf I} - H_B {\bf B}) + 2 k_1 {\bf A} - \widetilde{k}_2 ({\bf B}) n_B {\bf B}.  \\
  \end{cases}
\end{equation}
where $$\widetilde{k}_2  ({\bf B}) = \frac{n_B^{d/2}}{2^d (\pi)^{d/2} (\det ({\bf B})  )^{1/2}} k_2.$$
This is the model obtained by the ``variation-then-closure'', i.e., applying the maximum entropy closure at the PDE level. One can prove that the closure system (\ref{Direct-closure}) possesses an energy-dissipation law. To show this, we first look at the case with $\uvec = 0$. 
\begin{proposition}
In absence of the flow field $\uvec = 0$,  given $n_A > 0$, $n_B > 0$ and symmetric, positive-definite matrices ${\bf A}$ and ${\bf B}$, the closure system (\ref{Direct-closure}) satisfies the energy-dissipation law
  \begin{equation}\label{ED_cl_1}
    \frac{\dd}{\dd t} \mathcal{F}^{\rm CL}(n_A, n_B, {\bf A}, {\bf B}) = - \triangle^{\rm CL} \leq 0,
    \end{equation}
  where $\mathcal{F}^{\rm CL}(n_A, n_B, {\bf A}, {\bf B})$ is the coarse-grained free energy given by 
  \begin{equation}\label{E_closure}
    \begin{aligned}
   \mathcal{F}^{\rm CL} (n_A, n_B, {\bf A}, {\bf B})  & = \int n_A \left( \ln \left( \frac{n_A}{n_A^{\infty}}  \right) - 1 \right) + n_B \left( \ln \left( \frac{n_B}{n_B^{\infty}}  \right) - 1 \right) \\
   & - \frac{n_A}{2}   \ln \det \left( \frac{H_A {\bf A}}{n_A}  \right)  + \frac{1}{2}  \tr (H_A  {\bf A} - n_A {\bf I}  )  \\
   & - \frac{n_B}{2} \ln \det \left( \frac{H_B {\bf B}}{n_B}  \right)    + \frac{1}{2} \tr (H_B  {\bf B} -  n_B {\bf I}) )~  \dd \x, \\
    \end{aligned}
  \end{equation}
  and $\triangle^{\rm CL}$ is the rate of energy dissipation, given by
  \begin{equation}\label{D_closure}
    \begin{aligned}
  & \triangle^{\rm CL} = \int \xi_A \tr \left( (H_A {\bf I }- n_A {\bf A}^{-1})^2 {\bf A}  \right) + \xi_B \tr \left( (H_B {\bf I} - n_B {\bf B}^{-1})^2 {\bf B}  \right) \\
  & \quad  + (k_1 n_A  - \widetilde{k}_2 ({\bf B}) n_B^2) \left( \ln \left( \frac{n_A}{n_A^{\infty}} \right) -  2 \ln \left( \frac{n_B}{n_B^{\infty}}  \right) +  \ln \frac{\det  (H_B {\bf B} / n_B)}{\sqrt{ \det (H_A {\bf A}/ n_A) }} \right) \\
  & \quad  + \tr \left( (k_1 {\bf A} - \frac{1}{2}\widetilde{k}_2({\bf B}) n_B {\bf B}) (n_B {\bf B}^{-1} - \frac{1}{2}  n_A {\bf A}^{-1}) \right)  \dd \x. \\
    \end{aligned}
  \end{equation}
  In particular, under the condition that $n_A > 0$, $n_B > 0$, and ${\bf A}$ and ${\bf B}$ are symmetric positive-definite, $\triangle^{\rm CL} \geq 0$.
\end{proposition}

\begin{remark}
  The coarse-grained free energy $\mathcal{F}^{\rm CL} (n_A, n_B, {\bf A}, {\bf B}) $ is same to the macroscopic free energy given in \cite{germann2013nonequilibrium}. The free energy contains two part: the Oldroyd-B type elastic energy  associated  with  species $A$ and $B$ \cite{hu2007new,zhou2018dynamics}, and the Lyapunov function of the chemical reaction $\ce{A <=>[k_1][\widetilde{k}_2^{\rm eq}] 2 B}$ on number density with $k_1 n_A^{\infty} = \widetilde{k}_2^{eq} (n_B^{\infty})^2$ and $\widetilde{k}_2^{\rm eq} = H^{d/2}/(2^d \pi^{d/2})$. 
\end{remark}

\begin{proof}


We first show that we have the identity (\ref{ED_cl_1}) if $n_A$, $n_B$, ${\bf A}$ and ${\bf B}$ satisfy equation (\ref{Direct-closure}) with $\uvec = 0$. Indeed, for $\mathcal{F}^{\rm CL} (n_A, n_B, {\bf A}, {\bf B})$, a direct computation leads to 
\begin{equation}\label{dFdt}
  \begin{aligned}
 & \frac{\dd}{\dd t} \mathcal{F}^{\rm CL} = \frac{\dd}{\dd t} \int     n_A \left( \ln \left( n_A / n_A^{\infty}  \right) - 1 \right) + n_B \left( \ln \left( n_B / n_B^{\infty}  \right) - 1 \right) \\
 & \quad \quad - n_A  \ln \det \left( H_A {\bf A} / n_A  \right) / 2  +  \tr (H_A  {\bf A} - n_A {\bf I}  ) / 2 \\
 & \quad \quad - n_B \ln \det \left( H_B {\bf B} / n_B  \right) / 2   + \tr (H_B  {\bf B} -  n_B {\bf I}) )/ 2~  \dd \x. \\
 & = \int (\ln n_A - \ln n_A^{\infty} - \ln \det (H_A {\bf A}) / 2+ d \ln (n_A) / 2) \pp_t n_A  \\
 & \quad \quad +  (\ln n_B - \ln n_B^{\infty} - \ln \det (H_B {\bf B}) / 2+ d \ln (n_B) / 2) \pp_t n_A  \\
 & \quad \quad + \tr ( ( H_A {\bf I} - n_A {\bf A}^{-1} ) \pp_t {\bf A}) ) / 2 +  \tr ( ( H_B {\bf I} - n_B {\bf B}^{-1} ) \pp_t {\bf B}) ) / 2 \dd \x. \\
  \end{aligned}
\end{equation}
Substituting (\ref{Direct-closure}) into (\ref{dFdt}), and rearranging term, we have
\begin{equation*}
  \begin{aligned}
\frac{\dd}{\dd t} \mathcal{F}^{\rm CL} & =  - \int \xi_A \tr \left( (H_A {\bf I }- n_A {\bf A}^{-1})^2 {\bf A}  \right) + \xi_B \tr \left( (H_B {\bf I} - n_B {\bf B}^{-1})^2 {\bf B}  \right) \\
& \quad  + (k_1 n_A  - \widetilde{k}_2 ({\bf B}) n_B^2) \left( \ln \left( \frac{n_A}{n_A^{\infty}} \right) -  2 \ln \left( \frac{n_B}{n_B^{\infty}}  \right) +  \ln \frac{\det H_B ({\bf B} / n_B)}{\sqrt{ \det (H_A {\bf A}/ n_A) }} \right) \\
& \quad  + \tr \left( (k_1 {\bf A} - \widetilde{k}_2({\bf B}) n_B {\bf B} / 2) (n_B {\bf B}^{-1} -   n_A {\bf A}^{-1} / 2) \right)  \dd \x. \\  
  \end{aligned}
\end{equation*}

To prove $\triangle^{\rm CL} \geq 0$, we first define the quasi-equilibrium state $\psi_A^*$ and $\psi_B^*$ for given $n_A > 0$, $n_B > 0$ and symmetric, positive-definite matrices ${\bf A}$ and ${\bf B}$. The existence and uniqueness of $\psi_A^*$ and $\psi_B^*$ have been shown in proposition 3.1.
Notice that for
\begin{equation}
  \psi_{\alpha}^*(\qvec) = \frac{n_{\alpha}}{\sqrt{ (2 \pi)^d \det({\bf M})}} \exp(- \frac{1}{2} \qvec^{T} {\bf M}^{-1}_{\alpha} \qvec). \quad \psi^{\infty}_{\alpha} = \frac{n_{\alpha}^{\infty}}{\sqrt{ (2 \pi)^d H_{\alpha}^{-d}}} \exp(- \frac{1}{2} H_{\alpha} \qvec^{T}  \qvec),
\end{equation}
we have
\begin{equation*}\label{F_1}
  \begin{aligned}
    \int \psi_{\alpha}^* \ln \left( \frac{\psi_{\alpha}^*}{\psi^{\infty}_{\alpha}}  \right) \dd \qvec  & = \int \psi_{\alpha}^* \left( \ln \frac{n_{\alpha}}{n_{\alpha}^{\infty}} + \ln  \frac{1}{\sqrt{\det (H_{\alpha} {\bf M}_{\alpha}})}+\frac{1}{2} \left(  -  \qvec^{T} {\bf M}_{\alpha}^{-1} \qvec +  H_{\alpha} \qvec^{T} \qvec  \right)  \right) \dd \qvec \\
    &  =  n_{\alpha} \ln \frac{n_{\alpha}}{n_{\alpha}^{\infty}} - \frac{n_{\alpha}}{2}  \ln (\det (H_{\alpha} {\bf M}_{\alpha}) + \tr (H_{\alpha} n_{\alpha} {\bf M}_{\alpha} - n_{\alpha} {\bf I}),  \\
  \end{aligned}
\end{equation*}
and
\begin{equation}\label{D_1}
  \begin{aligned}
  \int \psi_{\alpha}^* \left| \nabla_{\qvec} \left( \ln \frac{\psi_{\alpha}^*}{\psi_{\alpha}^{\infty}}  \right) \right|^2 \dd \qvec & = \int \psi_{\alpha}  \left| \nabla_{\qvec} ( - \frac{1}{2} \qvec^{T} M_{\alpha}^{-1} \qvec + \frac{1}{2} H_{\alpha} \qvec^{T} \qvec) \right|^2  \dd \qvec \\
  & = \tr ( - {\bf M}_{\alpha}^{-1} + H_{\alpha} I )^2 n_{\alpha} {\bf M}_{\alpha}. \\
  \end{aligned}
\end{equation}
Moreover, by using the fact that $k_1 \psi_A^{\infty} = k_2 (\psi_B^{\infty})^2$,  we have 
\begin{equation}\label{D_2}
  \begin{aligned}
  & \int (k_1 \psi^*_A - k_2 (\psi_B^*)^{2}) \left(\ln  \frac{ \psi_A^*}{\psi_A^{\infty}} - 2 \ln  \frac{ \psi_B^*}{\psi_B^{\infty}}  \right) \dd \qvec \\
  & = \int (k_1 \psi^*_A - k_2 (\psi_B^*)^{2})  \Biggl[ \ln \frac{n_{A}}{n_{A}^{\infty}} + \ln  \frac{1}{\sqrt{\det (H_{A} \widetilde{\bf A}})}+\frac{1}{2} \left(  -  \qvec^{T} \widetilde{\bf A}^{-1} \qvec +  H_{A} \qvec^{T} \qvec  \right)    \\
  & \quad \quad \quad \quad \quad \quad \quad \quad  - 2 (\ln \frac{n_{B}}{n_{B}^{\infty}} + \ln  \frac{1}{\sqrt{\det (H_{B} \widetilde{\bf B}})}+\frac{1}{2} \left(  -  \qvec^{T} \widetilde{\bf B}^{-1} \qvec +  H_{B} \qvec^{T} \qvec) \right) \Biggr] \\
  & =  (k_1 n_A  - \widetilde{k}_2 ({\bf B}) n_B^2) \left( \ln \left( \frac{n_A}{n_A^{\infty}} \right) -  2 \ln \left( \frac{n_B}{n_B^{\infty}}  \right) +  \ln \frac{\det H_B \widetilde{B}}{\sqrt{ \det (H_A \widetilde{A} }} \right) \\
  & \quad \quad + \frac{1}{2} \tr \left(( - \widetilde{\bf A}^{-1} + H_A {\bf I} + 2 \widetilde{\bf B}^{-1} - 2 H_B {\bf I} )(k_1 {\bf A} -  \frac{1}{2} k_2 ({\bf B}) n_B \widetilde{{\bf B} })  \right),
  \end{aligned}
\end{equation}
where the last equality follows (\ref{close_n}) and (\ref{close_M}). Using $H_A =  2 H_B$ and combining the above calculations (\eqref{F_1}, \eqref{D_1} and \eqref{D_2}), we can show the (\ref{ED_cl_1}) is exactly same to  
\begin{equation}\label{ED_cl_2}
  \begin{aligned}
  & \frac{\dd}{\dd t} \int \int \psi_A^{*} \left(\ln \left( \frac{\psi_A^{*}} {\psi^{\infty}_{A}}  \right) - 1 \right) + \psi_B^{*} \left( \ln \left( \frac{\psi_B^{*}}{\psi^{\infty}_{B}} \right) - 1  \right) \dd \qvec \dd \x \\
 &  = - \int \int  \xi_A \psi_A \left|\nabla_{\qvec} \left(\ln \left( \frac{\psi_A^*}{\psi_A^{\infty}} \right)  \right) \right|^2  + \xi_B \psi_B \left|\nabla_{\qvec}  \left(\ln \left( \frac{\psi_B^*}{\psi_B^{\infty}} \right) \right) \right|^2 \\
 & \quad \quad \quad +  (k_1 \psi^*_A - k_2 (\psi_B^*)^{2}) \ln \left( \frac{k_1 \psi_A^*}{k_2 (\psi_B^*)^2}  \right) \dd \qvec \dd \x,
  \end{aligned}
\end{equation}
which is obtained by replacing $\psi_{\alpha}$ by $\psi_{\alpha}^*$ in the original micro-macro energy-dissipation law (\ref{Energy_Id}). It is clear that the right-hand side of (\ref{ED_cl_2}) is nonnegative, i.e., $\triangle^{\rm CL} \geq 0$.
\end{proof}

  With proposition 3.2, it is straightforward to show that the closure model (\ref{Direct-closure}) satisfies the energy-dissipation law
  \begin{equation}\label{ED_CL_u}
\frac{\dd}{\dd t} \left( \int \frac{1}{2} \rho |\uvec|^2 \dd \x + \mathcal{F}^{\rm CL} (n_A, n_B, {\bf A}, {\bf B} \right) = - \left( \int \eta |\nabla \uvec|^2 \dd \x + \triangle^{\rm CL} \right) \leq 0 
  \end{equation}
  for $n_A > 0$, $n_B > 0$ and  symmetric,  positive-definite matrices ${\bf A}$ and ${\bf B}$.
  However, it is not straightforward to derive the equation (\ref{Direct-closure}) from the the energy-dissipation law (\ref{ED_CL_u}). Moreover, due to presence of the reaction procedure, the dynamics (\ref{Direct-closure}) no longer lies on the quasi-equilibrium manifold $\mathcal{M}^*$.
  Indeed, the maximum entropy closure only use the information of the free energy part of the original system, it is unclear whether it is suitable for the dissipation part.
  As discussed in the next section, the closure model (\ref{Direct-closure}) fails to procedure a non-monotonic curve of the shear stress versus the applied shear rate in steady homogeneous flows. Such a closure approximation may only valid when the elastic part reaches its equilibrium much faster then the reaction part in the original system, i.e., the solution will move to $\mathcal{M}^*$ rapidly \cite{gorban2001corrections}. Unfortunately, in a high shear rate region, the macroscopic flow prevents the elastic part to reach its equilibrium.







\subsection{The moment closure model: closure-then-variation} 

To obtain a thermodynamically consistent macroscopic model that suitable for the high shear rate region, 
we consider a different closure approximation procedure, known as \emph{closure-then-variation}. The idea is to apply the closure approximation to the energy dissipation law first, and derive the closure system by applying the energetic variational approach in the coarse-grained level. This approach is similar to the Onsager principle based dynamic coarse graining method proposed in \cite{doi2016principle}.
By imposing a proper dissipation mechanism on the quasi-equilibrium manifold $\mathcal{M}^*$, we can have a thermodynamically consistent closure model for both mechanical and chemical part of the system.

On the quasi-equilibrium manifold $\mathcal{M}^*$, we have ${\bf A} = n_A \widetilde{\bf A}$ and ${\bf B} = n_B \widetilde{\bf B}$.
So the free energy $\mathcal{F}^{\rm CL} (n_A, n_B, {\bf A}, {\bf B} )$ for the closure system, defined in (\ref{E_closure}), can be reformulated in terms of number density $n_A$ and $n_B$, and the conformation tensor of two species $\widetilde{\bf A}$ and $\widetilde{\bf B}$, given by
\begin{equation}\label{E_closure_new}
  \begin{aligned}
 & \widetilde{\mathcal{F}}^{\rm CL} (n_A, n_B, \widetilde{\bf A}, \widetilde{\bf B} )  = \int n_A \left( \ln \left( \frac{n_A}{n_A^{\infty}}  \right) - 1 \right) + n_B \left( \ln \left( \frac{n_B}{n_B^{\infty}}  \right) - 1 \right) \\
 & \quad + \frac{n_A}{2} \left( - \ln \det \left( H_A \widetilde{{\bf A}} \right) + \tr \left( H_A \widetilde{\bf A} - {\bf I} \right)   \right) + \frac{n_B}{2} \left( - \ln \det \left( H_B \widetilde{{\bf B}} \right) + \tr \left( H_B \widetilde{{\bf B}} - {\bf I} \right) \right).  
  \end{aligned}
\end{equation}
We can impose the kinematics for the number density to account for the macroscopic breakage and combination procedure:
\begin{equation}
  \begin{cases}
& \pp_t n_A + \nabla \cdot (n_A \uvec) = - \pp_t R^n \\
& \pp_t n_B + \nabla \cdot (n_B \uvec) = 2 \pp_t R^n. \\
  \end{cases}
\end{equation}
where $R^n$ is the macroscopic reaction trajectory.

The dissipation of the macroscopic moment closure system on $\mathcal{M}^*$ consists of three parts: the viscosity of the macroscopic flow, the evolution of the conformation tensors and the reaction on the number density, which can be formulated as
\begin{equation}\label{Dissipation_on_manifold}
  \begin{aligned}
  \widetilde{\triangle}^* = \int &  \eta |\nabla \uvec|^2 
  + \tr \left( {\sf M}_A  \left(\frac{\dd \widetilde{\bf A}}{\dd t}\right)^2\right) + \tr \left( {\sf M}_B \left( \frac{\dd \widetilde{\bf B}}{\dd t} \right)^2 \right) ~ \dd \x + \widetilde{D}_{\rm chem} (R^n, \pp_t R^n) 
  \end{aligned} 
  \end{equation}
  where $\frac{\dd \bullet}{\dd t} = \pp_t \bullet + (\uvec \cdot \nabla) \bullet- (\nabla \uvec) \bullet - \bullet (\nabla \uvec)^{\rm T}$ is the kinematic transport of the conformation tensor \cite{lin2005hydrodynamics}, ${\sf M}_A(n_A, \widetilde{\bf A})$ and  ${\sf M}_B(n_B, \widetilde{\bf B})$ are mobility matrices. $\widetilde{D}_{\rm chem} (R^n, \pp_t R^n)$, defined by 
\begin{equation}
  \widetilde{D}_{\rm chem} (R^n, \pp_t R^n)  = \pp_t R^n \ln \left(  \eta^n( R^n) \pp_t R^n + 1  \right).
  \end{equation}
  is the dissipation for breakage and combination process at the macroscopic scale.
  The choice of $\eta^n(R^n)$ determines the macroscopic reaction rate in the closure system. One can view (\ref{Dissipation_on_manifold}) as a projection of the original dissipation on the quasi-equilibrium manifold $\mathcal{M}^*$. We then apply the energetic variational approach to obtain the dynamics on $\mathcal{M}^*$, i.e, the moment closure system.

\noindent {\bf The chemical reaction on the number density}: By performing energetic variational approach with respect to $R^n$ and $\pp_t R^n$, we obtain
\begin{equation}
  \ln \left(  \eta_n( R^n) \pp_t R^n + 1  \right) = - \frac{\delta \widetilde{\mathcal{F}}^{\rm CL}}{\delta R^n} = \mu_A^n -  2 \mu_B^n.
\end{equation}
For the closure energy  $\widetilde{\mathcal{F}}^{\rm CL}[n_A, n_B, \widetilde{\bf A}, \widetilde{\bf B}]$, we can compute the corresponding chemical potential of number density $n_A$ and $n_B$ as
\begin{equation}\label{mu_nAB}
  \begin{aligned}
    & \mu_A^n = \ln n_A - \ln n_A^{\infty} - \frac{1}{2} \ln \left( \det \left( H_A \widetilde{\bf A}  \right) \right)  + \frac{1}{2} \tr \left( H_A \widetilde{\bf A} - {\bf I} \right), \\
    & \mu_B^n = \ln n_B - \ln n_B^{\infty} - \frac{1}{2} \ln \left( \det \left( H_B \widetilde{\bf B} \right) \right)  + \frac{1}{2} \tr \left( H_B \widetilde{\bf B}- {\bf I} \right), \\
  \end{aligned}
\end{equation}
which is same to the generalized chemical potential defined in \cite{germann2013nonequilibrium}.

At the chemical equilibrium for given $\widetilde{\bf A}$ and $\widetilde{\bf B}$, $\mu_A^n = 2 \mu_B^n$, we have
\begin{equation}\label{Keq_ness}
K_{eq}^{\rm neq} = \frac{n_A^{\rm neq}}{ (n_B^{\rm neq})} = \frac{n_A^{\infty} \exp(- \frac{1}{2} \tr(\bm{\tau}_A/n_A)) \sqrt{\det \left( H_A \widetilde{\bf A}  \right)} }{ (n_B^{\infty})^2  \exp(- \tr(\bm{\tau}_B / n_B)) \det \left( H_B \widetilde{\bf B}  \right)},
\end{equation}
where 
\begin{equation}
\bm{\tau}_A = H_A {\bf A} - n_A {\bf I}, \quad \bm{\tau_B} =  H_B {\bf B} - n_B {\bf I}
\end{equation}
is the induced stress tensor associated with species $A$ and $B$ respectively.
Following \cite{germann2013nonequilibrium}, we take 
\begin{equation}
1/\eta^n(R^n) = \widetilde{k}_2  \exp( \tr (\bm{\tau}_B) / n_B) / \det (  H_B \widetilde{\bf B} ) n_B^2,
\end{equation}
which gives
\begin{equation}\label{Def_kneq}
  k_1^{\rm neq} = k_1^{eq} \frac{\exp( \frac{1}{2} \tr (\bm{\tau}_A / n_A) ) }{\sqrt{ \det (  H_A \widetilde{\bf A} )}} , \quad k_2^{\rm neq} = \widetilde{k}_2^{eq}  \frac{ \exp( \tr (\bm{\tau}_B / n_B)  }{ \det (  H_B \widetilde{\bf B} ) }.
\end{equation}
Thus, the number densities satisfy 
\begin{equation}\label{eq_n}
  \begin{aligned}
  & \pp_t n_{A} + \nabla \cdot (n_A \uvec ) = - k_1^{\rm neq} n_A + k_2^{\rm neq} n_B^2, \\
  & \pp_t n_{B} + \nabla \cdot (n_B \uvec ) = 2 k_1^{\rm neq} n_A -  2 k_2^{\rm neq} n_B^2, \\
  \end{aligned}
\end{equation}
The resulting non-equilibrium reaction rates of number densities is exact same to those in the GCB model \cite{germann2014investigation, germann2013nonequilibrium}.

\noindent {\bf Gradient flows with convection on conformation tensors}: The evolution of conformation tensor can be obtained by performing energetic variational approach in terms of ${\bf A}$ (${\bf B}$) and $\frac{\dd {\bf A}}{\dd t}$ ($\frac{\dd {\bf B}}{\dd t}$) \cite{Giga2017,zhou2018dynamics}, which result in  
\begin{equation}
  \begin{cases}
    & {\sf M}_A ( \pp_t \widetilde{\bf A} + (\uvec \cdot \nabla) \widetilde{\bf A} - (\nabla \uvec) \widetilde{\bf A} - \widetilde{\bf A} (\nabla \uvec)^{\rm T} ) = -   \dfrac{\delta \mathcal{F}^{\rm CL}}{\delta {\bf A}} \\
    & \\
    & {\sf M}_B (\pp_t \widetilde{\bf B} + (\uvec \cdot \nabla) \widetilde{\bf B} - (\nabla \uvec) \widetilde{\bf B} - \widetilde{\bf B} (\nabla \uvec)^{\rm T} ) = -   \dfrac{\delta \mathcal{F}^{\rm CL}}{\delta {\bf B}}. \\
  \end{cases}
\end{equation}
By taking ${\sf M}_A =  n_A {\widetilde{\bf A}^{-1}} / 4 \xi_A$ and ${\sf M}_B =  n_B {\widetilde{\bf B}^{-1}} / 4 \xi_B$, we have 
  \begin{equation}\label{nA_A}
    \begin{aligned}
      &  \pp_t \widetilde{\bf A} + (\uvec \cdot \nabla) \widetilde{\bf A} - (\nabla \uvec) \widetilde{\bf A} - \widetilde{\bf A} (\nabla \uvec)^{\rm T} = \xi_A (2 {\bf I} - 2 H_A \widetilde{\bf A}) \\
      &  \pp_t \widetilde{\bf B} + (\uvec \cdot \nabla) \widetilde{\bf B} - (\nabla \uvec) \widetilde{\bf B} - \widetilde{\bf B} (\nabla \uvec)^{\rm T} = \xi_B (2 {\bf I} - 2 H_B\widetilde{\bf B}). \\
    \end{aligned}
  \end{equation}

\noindent {\bf Macroscopic flow equation:} Now we compute the macroscopic flow equation by performing the energetic variational approach with respect to the flow map $\x(\X, t)$. When writing the macroscopic force balance, we should assume that the number densities and the conformation tensors to be purely transported with flow.
Under the incompressible condition ($\det \F =  1$), we have the kinematics \cite{lin2005hydrodynamics}
\begin{equation}
  \widetilde{\bf A} =  \F \widetilde{\bf A}_0 \F^{\rm T}, \quad \widetilde{\bf B} =  \F \widetilde{\bf B}_0 \F^{\rm T}, \quad n_A = n_A^0, \quad n_B  = n_B^0,
\end{equation}
and the action functional for the moment closure system can be given by
\begin{equation}
  \begin{aligned}
\widetilde{\mathcal{A}}[\x] = \int_{0}^T \int &  \frac{1}{2} \rho_0 |\x_t|^2 - \lambda \left[\frac{n_A^0}{2}  \tr \left( H_A \F \widetilde{\bf A}_0 \F^{\rm T}    \right) + \frac{n_B^0}{2}  \tr \left( H_B \F \widetilde{{\bf B}}_0 \F^{\rm T}\right) \right] \dd \X \dd \x \\
  \end{aligned}
\end{equation}
after dropping all the constant terms.
A direct computation results in
\begin{equation}
\frac{\delta \widetilde{\mathcal{A}}}{\delta \x} = - \rho (\uvec_t + \uvec \cdot \nabla \uvec) + \lambda \nabla \cdot (H_A {\bf A} + H_B {\bf B})).
\end{equation}
The only dissipation term for the macroscopic flow is the viscosity part $\mathcal{D}_{\eta} = \frac{1}{2}\int \eta |\nabla \uvec|^2 \dd \x$ \cite{Giga2017}, so the dissipative can be computed as $\frac{\delta \mathcal{D}_{\eta}}{\delta \x_t}  = - \eta \Delta \uvec$. The final macroscopic force balance can be written as
\begin{equation}\label{closure_u}
  \rho (\uvec_t + \uvec \cdot \nabla \uvec) + \nabla \tilde{p} = \eta \Delta \uvec + \lambda \nabla \cdot (H_A {\bf A} + H_B {\bf B})),
\end{equation}
where $\tilde{p}$ is a Lagrangian multiplier for the incompressible condition. 
Equation (\ref{closure_u}) is equivalent to
\begin{equation}\label{closure_u_1}
  \rho (\uvec_t + \uvec \cdot \nabla \uvec) + \nabla p = \eta \Delta \uvec + \lambda \nabla \cdot (H_A {\bf A} + H_B {\bf B} - (n_A + n_B) {\bf I}),
\end{equation}
due to the incompressible condition.

Finally, we get the the closure system
\begin{equation}\label{closure_system}
  \begin{cases}
    &  \rho (\pp_t \uvec + (\uvec \cdot \nabla) \uvec) + \nabla p = \eta \Delta \uvec + \lambda \nabla \cdot \left(  H_A n_A \widetilde{\bf A} + H_B n_B \widetilde{\bf B} - (n_A + n_B) {\bf I} \right)   \\
    & \nabla \cdot \uvec = 0 \\
    & \pp_t n_{A} + \nabla \cdot (n_A \uvec ) = - k_1^{\rm neq} n_A + k_2^{\rm neq} n_B^2, \\
    & \pp_t n_{B} + \nabla \cdot (n_B \uvec ) = 2 k_1^{\rm neq} n_A -  2 k_2^{\rm neq} n_B^2, \\
    &  \pp_t \widetilde{\bf A} + (\uvec \cdot \nabla) \widetilde{\bf A} - (\nabla \uvec) \widetilde{\bf A} - \widetilde{\bf A} (\nabla \uvec)^{\rm T} = 2 \xi_A ( {\bf I} -  H_A \widetilde{\bf A}) \\
    &  \pp_t \widetilde{\bf B} + (\uvec \cdot \nabla) \widetilde{\bf B} - (\nabla \uvec) \widetilde{\bf B} - \widetilde{\bf B} (\nabla \uvec)^{\rm T} = 2 \xi_B ( {\bf I} -  H_B\widetilde{\bf B}), \\
  \end{cases}
\end{equation}
where $k_1^{\rm neq}$ and $k_2^{\rm neq}$ are defined in (\ref{Def_kneq}). 
One can view (\ref{closure_system}) as a dynamics restricted in the quasi-equilibrium manifold $\mathcal{M}^*$. Recall that ${\bf A} = n_A \widetilde{\bf A}$ and ${\bf B} = n_B \widetilde{\bf B}$ on $\mathcal{M}^*$. Combining (\ref{nA_A}) with (\ref{eq_n}), we have the second moment equations
\begin{equation}\label{eq_SM}
  \begin{aligned}
  & \pp_t {\bf A} + (\uvec \cdot \nabla) {\bf A} - (\nabla \uvec) {\bf A} - {\bf A} (\nabla \uvec)^{\rm T} = 2 \xi_A (n_A {\bf I} - H_A {\bf A}) - k_1^{\rm neq} {\bf A} + k_2^{\rm neq} n_B^2 \widetilde{\bf A}\\
  & \pp_t {\bf B} + (\uvec \cdot \nabla) {\bf B} - (\nabla \uvec) {\bf B} - {\bf B} (\nabla \uvec)^{\rm T} =  2 \xi_B (n_B {\bf I} - H_B {\bf B}) + 2 k_1^{\rm neq} n_A \widetilde{\bf B} -  2 k_2^{\rm neq} n_B {\bf B}. \\
  \end{aligned}
\end{equation}
It is worth mentioning that the breakage and combination process actually create an active stress in the momentum equation if there exists additional mechanism to maintain the breakage and combination process away from an steady-state, as we can decompose $n_{\alpha}$ into two part, i.e., $n_{\alpha}(\x, t) = n_{\alpha}^{\infty}(\x) + n_{\alpha}^{\rm a}(\x, t)$ \cite{narayan2007long}. 
\begin{remark}
We notice that the reaction terms in the celebrate VCM \cite{vasquez2007network} and GCB models \cite{germann2013nonequilibrium} takes a different form. As mentioned in remark 2.2, the VCM model assume the microscopic reaction takes the form $k_1 \psi_A - k_2 \psi_B * \psi_B$ from (\ref{eq_SM}), which leads to the term $k_A {\bf A} - k_2 n_B {\bf B}$
in the second moment equation. The GCB model also take such a form as a starting point. 
To obtain the same form of reaction terms, one need further assume $2 \widetilde{\bf A} =  \widetilde{\bf B}$, then
\begin{equation}\label{App_AB}
\begin{aligned}
& - k_1^{\rm neq} {\bf A} + k_2^{\rm neq} n_B^2 \widetilde{\bf A} \approx - k_1^{\rm neq} {\bf A} + \frac{1}{2} k_2^{\rm neq} n_B{\bf B}, \\ 
& 2 k_1^{\rm neq} n_A \widetilde{\bf B} -  2 k_2^{\rm neq} n_B {\bf B} \approx 4 k_1^{\rm neq} {\bf A} -  2 k_2^{\rm neq} n_B {\bf B}. \\
\end{aligned}
\end{equation}
The assumption (\ref{App_AB}) is reasonable, since for given number densities $n_A$ and $n_B$, we have $\widetilde{ {\bf A}}^{eq} =\frac{1}{H_A} {\bf I}, \quad \widetilde{ {\bf B}}^{eq} =\frac{1}{H_B} {\bf I}$, which implies that $2 \widetilde{ {\bf A}}^{eq} =  \widetilde{ {\bf B}}^{eq}$ at the local equilibrium. So $2 \widetilde{\bf A} \approx  \widetilde{\bf B}$ is valid at least near the local equilibrium.
Under the approximation (\ref{App_AB}), we can reach a closure model
\begin{equation}\label{closure_system_1}
  \begin{cases}
    &\rho (\pp_t \uvec + (\uvec \cdot \nabla) \uvec) + \nabla p = \eta \Delta \uvec + \lambda \nabla \cdot \left(  H_A {\bf A} + H_B {\bf B} - (n_A + n_B) {\bf I} \right)   \\
    & \nabla \cdot \uvec = 0 \\
    &\pp_t n_{A} + \nabla \cdot (n_A \uvec ) = - k_1^{\rm neq} n_A + k_2^{\rm neq} n_B^2, \\
    &\pp_t n_{B} + \nabla \cdot (n_B \uvec ) = 2 k_1^{\rm neq} n_A -  2 k_2^{\rm neq} n_B^2, \\
    &\pp_t {\bf A} + (\uvec\cdot\nabla) {\bf A} - (\nabla \uvec) {\bf A} - {\bf A} (\nabla \uvec)^{\rm T} = 2 \xi_A (n_A {\bf I} - H_A {\bf A})  - k_1^{\rm neq} {\bf A} + \frac{1}{2}k_2^{\rm neq}  n_B {\bf B} \\
    &\pp_t {\bf B} + (\uvec\cdot\nabla) {\bf B} - (\nabla \uvec) {\bf B} - {\bf B} (\nabla \uvec)^{\rm T} =  2 \xi_B (n_B {\bf I} - H_B {\bf B}) + 4 k_1^{\rm neq} {\bf A}  - 2 k_2^{\rm neq} n_B {\bf B},  \\
  \end{cases}
\end{equation}
which has the same form of the VCM and GCB models. 
Although the dynamics (\ref{closure_system_1}) no longer lies on the quasi-equilibrium, it can produce more reasonable shear-stress curve at in a high shearing rate region. Compare with (\ref{closure_system}), (\ref{closure_system_1}) can force $\big|2 \widetilde{\bf A} - \widetilde{\bf B}\big|$ to be small due to the  approximation (\ref{App_AB}). We will compare these two models in details in the future work.
\end{remark}

\begin{remark}
  In the above derivation, we assume the second moments can be written as the multiplication of number density and the conformation tensor, i.e. ${\bf A} = n_A \widetilde{\bf A}$ and ${\bf B} = n_B \widetilde{\bf B}$. Such a decomposition is valid on the submanifold formed by the quasi-equilibrium states, but may not true in general. A different moment closure system can be obtained if one treat the number densities and the second moments to be independent.  Then the free energy of the closure system (\ref{E_closure}) can be written as
  \begin{equation}\label{E_closure1}
    \begin{aligned}
   \mathcal{F}^{\rm CL}(n_A, n_B, {\bf A}, {\bf B}) = \int &  n_A \left( \ln \left( \frac{n_A}{n_A^{\infty}}  \right) - \frac{1}{2} \det {H_A \bf A}  - 1 \right) + n_B \left( \ln \left( \frac{n_B}{n_B^{\infty}} \right) - 1 - \frac{1}{2} \det (H_B {\bf B}) \right) \\
   & + \frac{d}{2} (n_A \ln n_A - n_A + n_B \ln n_B - n_B) + \frac{1}{2} \tr(H_A {\bf A}) + \frac{1}{2} \tr (H_B {\bf B}), \\   
    \end{aligned}
  \end{equation}
  which implies that
  \begin{equation}
    \mu_A^n  = \ln n_A - \ln n_A^{\infty} - \frac{1}{2} \det (H_A {\bf A} / n_A), \quad  \mu_B^n  = \ln n_B - \ln n_B^{\infty} - \frac{1}{2} \det (H_B {\bf B} / n_B).
  \end{equation}
  We can simply modify the reaction rate in (\ref{Direct-closure}) by
  \begin{equation}
    k_1^{\rm neq} = k_1^{eq} / \det(H_A {\bf A} / n_A), \quad k_2^{\rm neq} = k_2^{eq} / \det (H_B {\bf B} / n_B)
  \end{equation}
  to obtain another closure model. We'll explore this in the future.
\end{remark}

\begin{remark}
  It is worth mentioning that the derivation in this section can be viewed as a pure macroscopic approach to model wormlike micellar solutions in the framework of EnVarA, which starts with the free energy $\mathcal{F}[n_A, n_B, \widetilde{\bf A}, \widetilde{\bf B}]$ and the dissipation given by (\ref{E_closure_new}) and (\ref{Dissipation_on_manifold}) respectively. As a pure macroscopic approach, it is not necessary to assume $H_A = 2 H_B$. If we assume $H_B = 2 H_A$, and adopt the approximation $ {{\bf A}} / n_A \approx 2 {{\bf B}} / n_B$,
  then we have
  \begin{equation}
    \begin{aligned}
    & - k_1^{\rm neq} {\bf A} + k_2^{\rm neq} n_B^2 \widetilde{\bf A} \approx - k_1^{\rm neq} {\bf A} + 2 k_2^{\rm neq} n_B{\bf B}, \\ 
    & 2 k_1^{\rm neq} n_A \widetilde{\bf B} -  2 k_2^{\rm neq} n_B {\bf B} \approx k_1^{\rm neq} {\bf A} -  2 k_2^{\rm neq} n_B {\bf B}, \\
    \end{aligned}
  \end{equation}
  which is exactly same to those in the GCB model \cite{germann2013nonequilibrium,germann2014investigation}. 
\end{remark}

\section{Numerics}
In this section, we discuss the prediction of the above moment closure models through a few toy examples. 
Detailed numerical studies for the original micro-macro model and the closure models will be carried out in future work.

\subsection{Steady homogeneous shear flow}

First we consider a steady homogeneous shear flow with the velocity field given by $$\uvec = (\kappa y, 0),$$ where $\kappa$ is the constant shear rate.  Moreover, we assume that the number density of each species is spatial homogeneous. So the original PDE system is reduced to an ODE system of $n_{\alpha}$, ${\bf A}$ and ${\bf B}$. We solve the ODE system by the standard explicit Euler scheme. 
We take the initial condition as
\begin{equation}
n_A^0 = 1, \quad n_B^0 = 2.5, \quad {\bf A}_0 = \frac{n_A^0}{H_A} {\bf I}, \quad {\bf B}_{0} = \frac{n_B^0}{H_B} {\bf I}.
\end{equation}
For each $\kappa$, we compute the number density of each species, 
and the induced shear stress $ \tau_{12} = H_A A_{12} + H_B B_{12}$. We compare the predictions for three models (\ref{Direct-closure}), (\ref{closure_system}) and (\ref{closure_system_1}).
The terminal  criterion for the numerical calculation is $T = 2$ or $n_A < 10^{-5}$.
\begin{figure}[!h]
  \centering

\includegraphics[width = \linewidth]{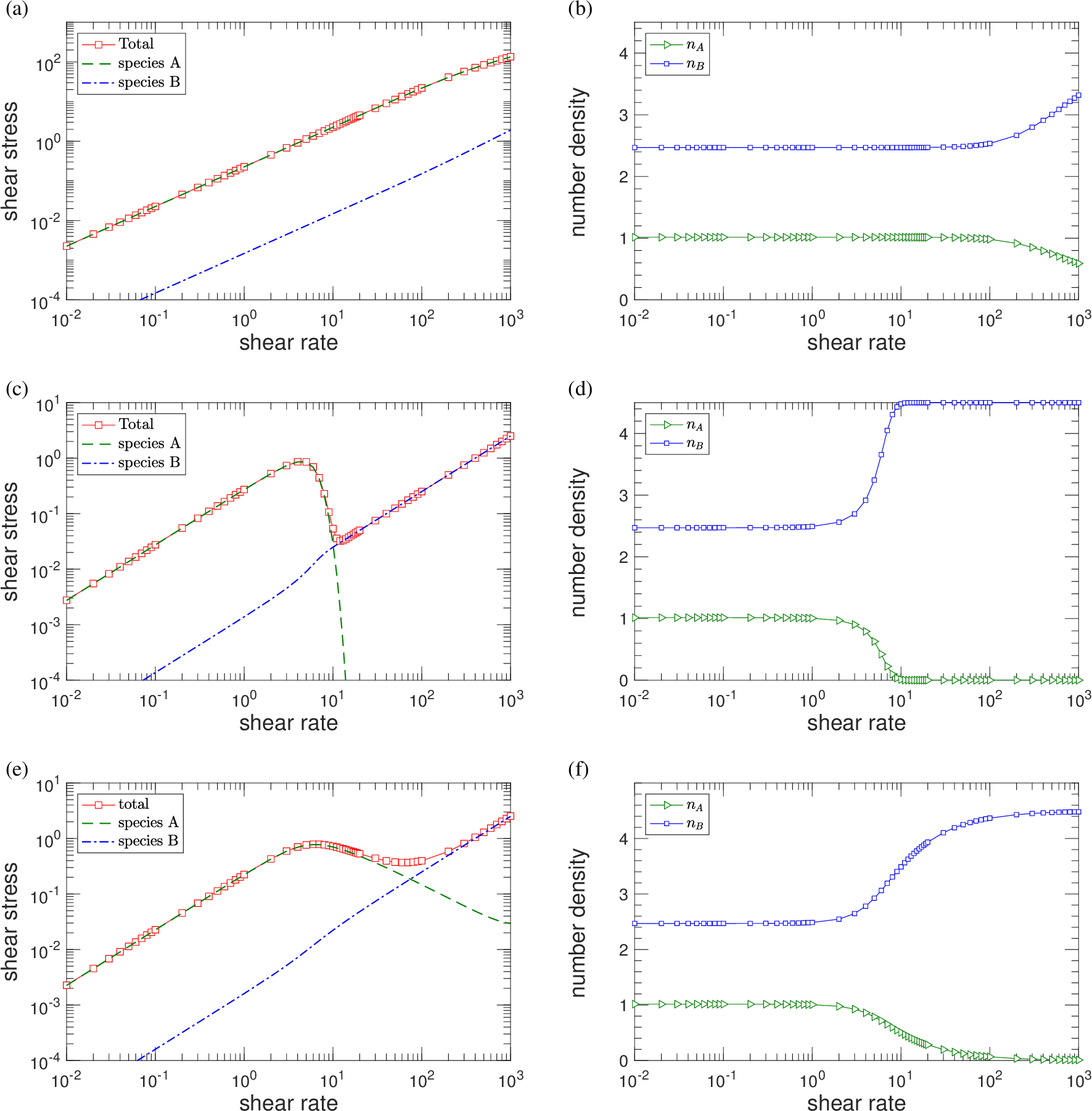}


\caption{Calculated shear stress and species number density as a function of $\kappa$ ($k_1^{eq} =0.9$, $\widetilde{k}_2^{eq} = 0.15$, $\xi_A = 0.9$, $\xi_A/\xi_B = 10^{-3}$, $H_A = 2$, $H_B = 1$). (a)-(b) Model (\ref{Direct-closure}); (c)-(d) Model (\ref{closure_system}); (e)-(f) Model (\ref{closure_system_1}). }\label{shear_stress}
 \end{figure}
Fig. \ref{shear_stress} showed the calculated shear stress and the number densities of two species as a function of $\kappa$ for three models ($k_1^{eq} =0.9$, $\widetilde{k}_2^{eq} = 0.15$, $\xi_A = 0.9$, $\xi_A/\xi_B = 10^{-3}$, $H_A = 2$, $H_B = 1$). 
At small shear rates, all three models can produce similar results, due to the fact that $\widetilde{\bf A}$ and $\widetilde{\bf B}$ are close to their equilibria. The closure model (\ref{Direct-closure}), obtained by applying the maximum entropy closure to the equation directly, fails to obtain a non-monotonic shear-stress curve. The main reason might be the fact that the break rate $k_1$ is independent with the shear rate in this model, which can not lead to a pronounced breakage of species $A$. The predictions of model (\ref{closure_system}) and (\ref{closure_system_1}) are also different in the high shear rate region. The model (\ref{closure_system}) leads to a rapidly breakage of species $A$ (Fig. \ref{shear_stress} (c) - (d)), which does not seem to match previous experimental and simulation results \cite{germann2013nonequilibrium}. 
The curves produced by the model \ref{closure_system_1}, shown in Fig. \ref{shear_stress}(e)-(f), is consistent with the results by the VCM and GCB models qualitatively \cite{germann2013nonequilibrium}. 
As mentioned earlier, the approximation (\ref{App_AB}) can be viewed as an implicit regularization term such that $|2 \widetilde{\bf A} - \widetilde{\bf B}|$ to be small, which prevent $\widetilde{A}_{12}$ to be too large.
This simple numerical test shows the importance of choosing a proper dissipation in the course-grained level in order to capture the non-equilibrium rheological properties of wormlike micellar solutions. A detailed comparison of different closure models will be made in future work.

\subsection{Transient behavior in a planar shear flow}

In this subsection, we investigate the transient behavior of the model in a planar shear flow for the closure model (\ref{closure_system_1}). Let $\uvec = (u(y), 0)$ and $u(y)$ satisfies
\begin{equation}
  \begin{cases}
    & u_t  = \eta \pp_{yy} u + \lambda \pp_y (H_A A_{12} + H_B B_{12}), \\
    & u(l) = \kappa(t), \quad u(0) = 0, \\
  \end{cases}
\end{equation}
we take $\kappa(t) = \gamma \tanh(a t)$, where $a$ is a parameter control how the wall velocity approaches steady-state \cite{zhou2012multiple, germann2014investigation}. Other parameters are set as: $l = 0.1$, $H_A = 2$, $H_B = 1$, $\xi_A = 0.9$, $\xi_A/\xi_B = 10^{-3}$, $\eta = 1$, $\lambda = 1$, $k_1 = 1$ and $\widetilde{k}_2^{\rm eq} = 6.25$. The numerical setup is close to the case considered in \cite{germann2014investigation}, but we consider the Couette flow between two surface instead of Taylor-Couette flow in the gap between two rotating cylinders for simplicity. We fix $\gamma = 50$ through this subsection.

\begin{figure}[!htbp]

\includegraphics[width = \linewidth]{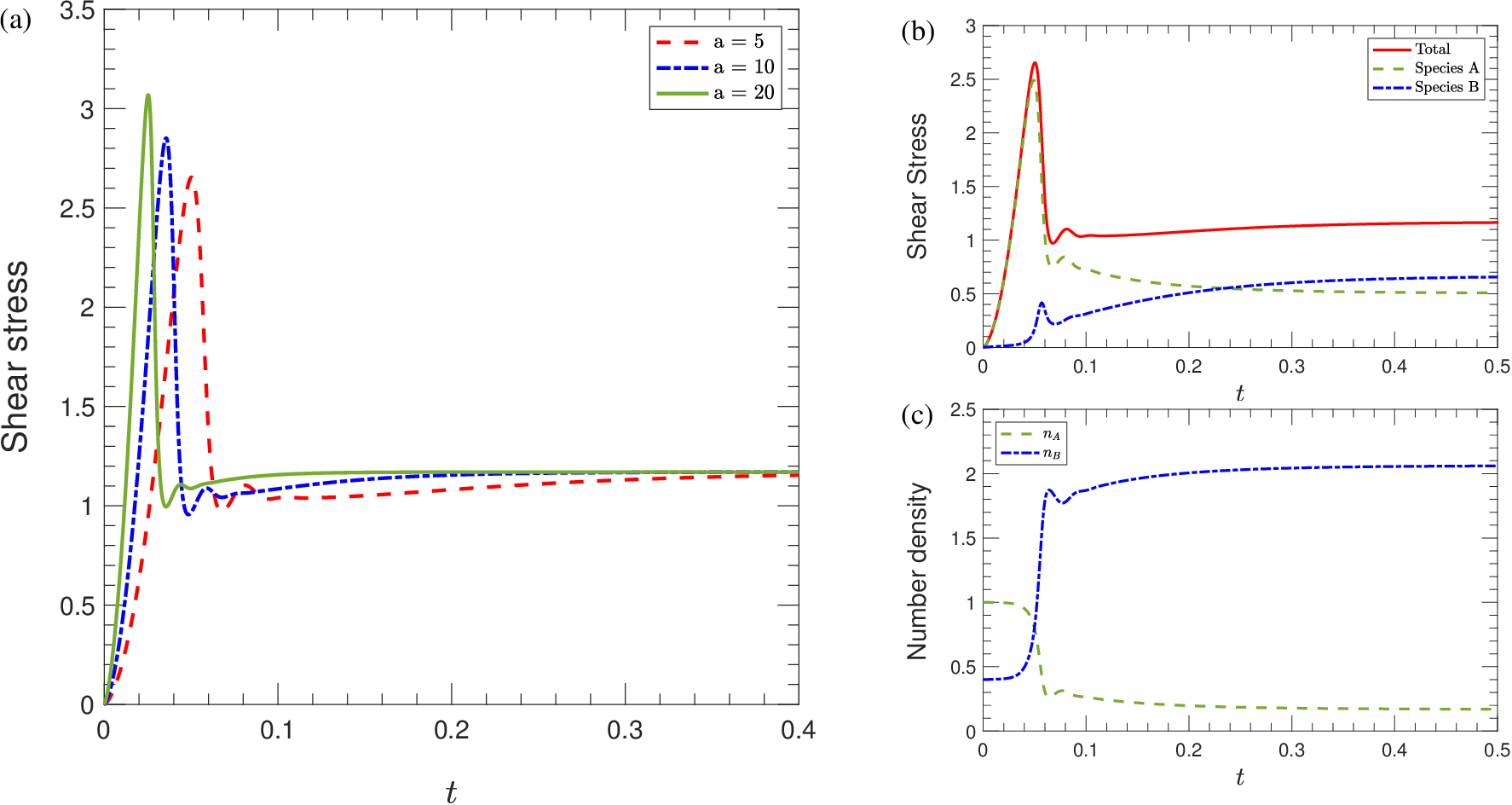}
  \caption{Transient behavior of the closure model (\ref{closure_system_1}) in a planar Couette flow: (a) Calculated shear stress at moving wall for different ramp-up rate $a$. (b) The wall shear stress as a function of $t$, (c) Species number density at the wall.}\label{Transient}
\end{figure}

Fig. \ref{Transient}(a) shows the transient response of the wall shear stress tensor for different ramp-up rates $a$. In all three cases, the shear stress will reach its maximum during the ramp-up process. Different ramp-up rates do not significantly affect the steady-state.
Fig. \ref{Transient} (b) shows temporal evolution of the total stress at the moving surface for $\kappa(t) = 50 \tanh (5t)$, the individual contributions of species of $A$ and $B$ are represented by dashed and dash-dotted lines. The number densities of species $A$ and $B$ are plotted  in Fig. \ref{Transient} (c). The above results are qualitatively agree with rheological characteristics predicted by the GCB model in circular Taylor-Couette flow (see Fig. 4 and Fig. 5 in \cite{germann2014investigation}).

\section{Summary}
In this paper, inspired by the celebrated VCM type models \cite{vasquez2007network,germann2013nonequilibrium}, we derive a thermodynamically consistent two-species micro-macro model of wormlike micellar solutions by employing an energetic variational approach. Our model incorporates a breakage and combination process of polymer chains into a classical micro-macro dumbbell model for polymeric fluids in a unified variational framework. 
The energetic variational formulation for the micro-macro model opens a new door for both numerical studies and theoretical analysis \cite{Liu-Wang-Zhang_WLMs}.  
The modeling approach also provides a framework to integrate other mechanism, and can be applied to other chemo-mechanical systems beyond the wormlike micellar solutions, such as active soft matter systems \cite{cifre2003brownian, narayan2007long, prost2015active, tiribocchi2015active, yang2016hydrodynamic}. %


We also study the maximum entropy closure approximation to the micro-macro model of wormlike micellar solutions. The maximum entropy closure links the micro-macro model with the VCM-type macroscopic model \cite{vasquez2007network, grmela2010mesoscopic, germann2013nonequilibrium}. 
We compare closure approximations by both ``variation-then-closure'' and ``closure-then-variation'' approaches. 
We show that these two approaches result in different closure models due to presence of the chemical reaction.
Since maximum entropy closure only uses the information from the free energy part of the original system \cite{gorban2001corrections}, applying the closure approximation on the PDE level cannot guarantee the thermodynamical consistency.
By a ``closure-then-variation'' approach, we can restrict the dynamics on the coarse-grained manifold by choosing the dissipation properly. As a consequence, the closure system preserves  
the thermodynamical structures of the original system for both chemical and mechanical parts. Several numerical examples show that the closure model, obtained by ``closure-then-variation'' can capture the key rheological features of wormlike micellar solution. The variational structures of models in both levels are crucial for the stability of whole system and the accuracy of structure-preserving numerical simulations \cite{liu2020lagrangian, liu2020variational, vermeeren2019contact}.
A detailed numerical study for our models will be carried out in future work.

\section*{Acknowledgement}
Y. Wang and C. Liu are partially supported by the National Science Foundation (USA) NSF DMS-1950868 and the United States-Israel Binational Science Foundation (BSF) \#2024246. T-F. Zhang is partially supported by the National Natural Science Foundation of China No. 11871203. This work was done when T.-F. Zhang visited Illinois Institute of Technology during 2019-2020, he would like to acknowledge the sponsorship of the China Scholarship Council, under the State Scholarship Fund (No. 201906415023) and the hospitality of Department of Applied
Mathematics at Illinois Institute of Technology. The authors would like to thank Prof. Haijun Yu for suggestions and helpful discussions.

\section*{References}
\bibliographystyle{cas-model2-names}
\bibliography{MM}

\end{document}